\documentclass[11pt]{article}

\usepackage{amsthm}
\usepackage{graphicx} 
\usepackage{array} 

\usepackage{amsmath, amssymb, amsfonts, verbatim}
\usepackage{hyphenat,epsfig,subcaption,multirow}
\usepackage{nicefrac}
\usepackage{paralist}
\usepackage{thm-restate}
\usepackage{mleftright}

\usepackage[font=small, labelfont=bf]{caption}

\usepackage[usenames,dvipsnames]{xcolor}
\usepackage[ruled]{algorithm2e}

\DeclareFontFamily{U}{mathx}{\hyphenchar\font45}
\DeclareFontShape{U}{mathx}{m}{n}{
      <5> <6> <7> <8> <9> <10>
      <10.95> <12> <14.4> <17.28> <20.74> <24.88>
      mathx10
      }{}
\DeclareSymbolFont{mathx}{U}{mathx}{m}{n}
\DeclareMathSymbol{\bigtimes}{1}{mathx}{"91}

\usepackage{tcolorbox}
\tcbuselibrary{skins,breakable}
\tcbset{enhanced jigsaw}

\usepackage[normalem]{ulem}
\usepackage[compact]{titlesec}

\definecolor{DarkRed}{rgb}{0.5,0.1,0.1}
\definecolor{DarkBlue}{rgb}{0.1,0.1,0.5}

\usepackage{nameref}
\definecolor{ForestGreen}{rgb}{0.1333,0.5451,0.1333}
\definecolor{Red}{rgb}{0.9,0,0}
\usepackage[linktocpage=true,
	pagebackref=true,colorlinks,
		urlcolor=black,
	linkcolor=DarkRed,citecolor=ForestGreen,
	bookmarks,bookmarksopen,bookmarksnumbered]
	{hyperref}
\usepackage[noabbrev,nameinlink]{cleveref}
\crefname{property}{property}{Property}
\creflabelformat{property}{(#1)#2#3}
\crefname{equation}{eq}{Eq}
\creflabelformat{equation}{(#1)#2#3}

\usepackage{bm}
\usepackage{url}
\usepackage{xspace}
\usepackage[mathscr]{euscript}

\usepackage{environ}  
\usepackage{chngcntr} 

\usepackage{tikz}
\usetikzlibrary{arrows}
\usetikzlibrary{arrows.meta}
\usetikzlibrary{shapes}
\usetikzlibrary{backgrounds}
\usetikzlibrary{positioning}
\usetikzlibrary{decorations.markings}
\usetikzlibrary{patterns}
\usetikzlibrary{calc}
\usetikzlibrary{fit}

\tikzset{vertexA/.style={circle,fill=black,minimum size=5pt,inner sep=0pt, font=\small}}
\tikzstyle{vertexB}=[circle,fill=black,minimum size=3pt,inner sep=0pt, font=\small]
\tikzstyle{selected vertex} = [vertex, fill=red!24]
\tikzstyle{edge} = [draw, black!50, -]
\tikzstyle{vertexRed}=[circle,fill=red,draw,minimum size=10pt,inner sep=0pt, font=\small, line width=0.5pt]
\tikzstyle{vertexBlack}=[circle,fill=black,draw,minimum size=10pt,inner sep=0pt, font=\small, line width=0.5pt]

\usepackage{mdframed}

\usepackage[noend]{algpseudocode}
\makeatletter
\def\BState{\State\hskip-\ALG@thistlm}
\makeatother

\usepackage{cite}
\usepackage{enumitem}

\usepackage[margin=1in]{geometry}

\newtheorem{lemma}{Lemma}[section]

\newtheorem{claim}[lemma]{Claim}

\newtheorem*{claim*}{Claim}
\newtheorem*{proposition*}{Proposition}
\newtheorem*{lemma*}{Lemma}
\newtheorem*{problem*}{Problem}

\crefname{lemma}{Lemma}{Lemmas}
\crefname{claim}{Claim}{Claims}

\newtheorem{mdresult}{Result}

\newtheorem*{mdresult*}{Main Result}

\newtheorem{assumption}{Assumption}

\theoremstyle{definition}

\newtheorem{mdinvariant}[lemma]{Lemma}

\theoremstyle{definition}
\newtheorem{mdalg}{Algorithm}
\newenvironment{Algorithm}{\begin{tbox}\begin{mdalg}}{\end{mdalg}\end{tbox}}

\allowdisplaybreaks

\renewcommand{\qed}{\nobreak \ifvmode \relax \else
      \ifdim\lastskip<1.5em \hskip-\lastskip
      \hskip1.5em plus0em minus0.5em \fi \nobreak
      \vrule height0.75em width0.5em depth0.25em\fi}

\renewcommand{\leq}{\leqslant}
\renewcommand{\geq}{\geqslant}

\newcommand{\eps}{\ensuremath{\varepsilon}}

\newcommand{\card}[1]{\left\vert{#1}\right\vert}

\newcommand{\set}[1]{\ensuremath{\left\{ #1 \right\}}}

\newenvironment{tbox}{\begin{tcolorbox}[
		enlarge top by=5pt,
		enlarge bottom by=5pt,
		 breakable,
		 boxsep=0pt,
                  left=4pt,
                  right=4pt,
                  top=10pt,
                  arc=0pt,
                  boxrule=1pt,toprule=1pt,
                  colback=white
                  ]
	}
{\end{tcolorbox}}

\newcommand{\II}{\ensuremath{\mathbb{I}}}

\newcommand{\mireal}[1][]{
  \ifx\relax#1\relax%
    \II(\mione \,; \mitwo)%
  \else%
    \II(\mione \,; \mitwo\mid #1)%
  \fi
}

\NewEnviron{constraint}{
	\begin{tcolorbox}
		\begin{equation}
			\BODY
		\end{equation}
	\end{tcolorbox}
}

\newcommand{\e}[1]{\eps_{#1}}

\newcommand{\anew}{\ensuremath{A_{new}}}
\newcommand{\bnew}{\ensuremath{B_{new}}}
\newcommand{\cnew}{\ensuremath{C_{new}}}
\newcommand{\pnew}{\ensuremath{P_{new}}}

\newcommand{\ao}{\ensuremath{A_{old}}}
\newcommand{\bo}{\ensuremath{B_{old}}}
\newcommand{\co}{\ensuremath{C_{old}}}
\newcommand{\po}{\ensuremath{P_{old}}}

\newcommand{\ebef}{\ensuremath{E_{bef}}}
\newcommand{\eaft}{\ensuremath{E_{aft}}}
\newcommand{\abef}{\ensuremath{A_{bef}}}
\newcommand{\aaft}{\ensuremath{A_{aft}}}
\newcommand{\bbef}{\ensuremath{B_{bef}}}
\newcommand{\baft}{\ensuremath{B_{aft}}}
\newcommand{\cbef}{\ensuremath{C_{bef}}}
\newcommand{\caft}{\ensuremath{C_{aft}}}

\newcommand{\To}{Type-1\xspace}
\newcommand{\Tt}{Type-2\xspace}

\newif\iffinaltheorem
\finaltheoremtrue 

\title{An Improved Fully Dynamic Algorithm for Counting 4-Cycles in General Graphs using Fast Matrix Multiplication}
\author{
	Sepehr Assadi\footnote{(\href{mailto:sepehr@assadi.info}{sepehr@assadi.info)} 
Cheriton School of Computer Science, University of Waterloo. Supported in part by a Sloan Research 
Fellowship, an NSERC Discovery Grant, and a Faculty of Math Research Chair grant.} \and 
Vihan Shah\footnote{(\href{mailto:vihan.shah@uwaterloo.ca}{\text{vihan.shah@uwaterloo.ca}}) 
Cheriton School of Computer Science, University of Waterloo. Supported in part by Sepehr Assadi's NSERC Discovery Grant.}  
}

\date{}

\begin{document}
\maketitle

\thispagestyle{empty}
\pagenumbering{roman}


\begin{abstract}
We study subgraph counting over fully dynamic graphs, which undergo edge insertions and deletions. Counting subgraphs is a fundamental problem in graph theory with numerous applications across various fields, including database theory, social network analysis, and computational biology. 

Maintaining the number of triangles in fully dynamic graphs is very well studied and has an upper bound of $O(m^{1/2})$ for the update time by Kara, Ngo, Nikolic, Olteanu, and Zhang (TODS 20). There is also a conditional lower bound of approximately $\Omega(m^{1/2})$ for the update time by Henzinger, Krinninger, Nanongkai, and Saranurak (STOC 15) under the OMv conjecture implying that $O(m^{1/2})$ is the ``right answer'' for the update time of counting triangles.
More recently, Hanauer, Henzinger, and Hua (SAND 22) studied the problem of maintaining the number of $4$-cycles in fully dynamic graphs and designed an algorithm with $O(m^{2/3})$ update time which is a natural generalization of the approach for counting triangles. Thus, it seems natural that $O(m^{2/3})$ might be the correct answer for the complexity of the update time for counting $4$-cycles. 

In this work, we present an improved algorithm for maintaining the number of $4$-cycles in fully dynamic graphs. Our algorithm achieves a worst-case update time of $O(m^{2/3-\eps})$ for some constant $\eps>0$. Our approach crucially uses fast matrix multiplication and leverages recent developments therein to get an improved runtime. Using the current best value of the matrix multiplication exponent $\omega=2.371339$ we get $\eps=0.009811$ and if we assume the best possible exponent i.e. $\omega=2$ then we get $\eps=1/24$.
The lower bound for the update time is $\Omega(m^{1/2})$, so there is still a big gap between the best-known upper and lower bounds. The key message of our paper is demonstrating that $O(m^{2/3})$ is not the correct answer for the complexity of the update time.

\end{abstract}

\clearpage
\bigskip
\setcounter{tocdepth}{3}
\tableofcontents

\clearpage

\pagenumbering{arabic}
\setcounter{page}{1}


\section{Introduction}\label{sec:intro} 

We study subgraph counting over fully dynamic graphs that undergo edge 
insertions and deletions.
Given an input graph $G$ and a small pattern graph $H$, say a triangle or a $4$-cycle, we are interested in 
counting the number of subgraphs (not necessarily induced) of $G$ that are isomorphic to $H$. 
Counting subgraphs is a fundamental problem in graph theory with numerous applications across various fields, including social network analysis, telecommunication networks, and 
computational biology (see, e.g.,
\cite{ribeiro2021survey,costa2011analyzing,milo2002network,sahu2020ubiquity}
 for some applications).

Subgraph counting has many applications in database theory \cite{ngo2014skew,NgoOptimalJoin,DBLP:journals/tods/KaraNNOZ20} as well.
The problem of counting the number of elements in binary joins over binary relations is equivalent to the problem of subgraph counting in \textbf{layered graphs} (formally defined in \Cref{subsec:prelim-notation}), where the vertex set is partitioned into multiple layers, each corresponding to an attribute. Vertices represent attribute values, and edges exist only between specific layers (typically consecutive ones), representing tuples of attribute values in the relations. 

Specifically, the problem of finding the number of elements in a cyclic join of size $k$ is equivalent to counting the number of $k$-cycles in $k$-layered graphs.
For example, let $R, S,$ and $T$ be relations that have schemas $(A, B),$ $(B, C),$ and $(C, A)$ respectively. 
Then the size of the join of $R$ with $S$ with $T$ equals the number of triangles in the corresponding layered graph where there is a layer for each attribute $A,B,C$, the vertices are the attribute values and the edges only exist between neighboring layers representing the tuples of attribute values in the relations.
Throughout this paper, we adopt the language of subgraph counting; however, all results can equivalently be rephrased in terms of joins.
An example of the connection between joins and layered graphs can be found in \Cref{fig:join-layered}.

In the dynamic model, the pattern graph $H$ is fixed and known apriori, and the input graph $G$ starts as an empty graph which then undergoes edge insertion and deletions.
If both insertions and deletions are allowed, the model is referred to as fully dynamic.
After \emph{every} update we need to output the \emph{exact} number of occurrences of the subgraph $H$ inside $G$.
The goal is to minimize the worst-case update time over all edge updates.

In this paper, we are interested in counting $4$-cycles in the fully dynamic model. 
We can equivalently state the problem as follows:
Given four binary relations $A(L_1,L_2)$, $B(L_2,L_3)$, $C(L_3,L_4)$, and $D(L_4,L_1)$ that undergo continuous updates through tuple insertions and deletions, the problem is to dynamically maintain the count of tuples in the cyclic join $ J = A \Join B \Join C \Join D $. 
While $J$ itself is not materialized, a view storing its count is maintained and must be updated after every update to the database.
As updates occur, incremental view maintenance (IVM) is used to efficiently update the join size by computing the change in count ($\Delta J$) rather than recomputing the join size from scratch. 
The goal is to minimize the worst-case update time over all updates to the database while ensuring that the updated count is output after each update.

We are interested in $4$-cycles because it is the smallest pattern graph for which the dynamic counting problem is not well-understood.
Counting subgraphs on $1$ and $2$ vertices is trivial and counting subgraphs on $3$ vertices is well studied.
In particular, $3$-cycle (triangle) counting in the dynamic model has an upper bound of $O(\sqrt{m})$ for the 
worst-case update time \cite{DBLP:journals/tods/KaraNNOZ20}.
There is also a conditional lower bound of $\Omega(m^{1/2-\gamma})$ for any constant $\gamma>0$ for the update 
time 
\cite{henzinger2015unifying} under the Online Matrix-Vector (OMv) conjecture.

However, the bounds for graphs on $4$ vertices are not well understood.
In particular, $4$-cycle counting in the dynamic model has an upper bound of $O(m^{2/3})$ 
\cite{hanauer_et_al:LIPIcs.SAND.2022.18} for the 
worst-case update time but the conditional lower bound is the same bound of $\Omega(m^{1/2-\gamma})$ for 
any constant $\gamma>0$
\cite{henzinger2015unifying,hanauer_et_al:LIPIcs.SAND.2022.18}.
\cite{hanauer_et_al:LIPIcs.SAND.2022.18} also studied the problem of counting $4$-cliques showing that the 
folklore upper bound of $O(m)$ for the update time is tight under the static combinatorial $4$-clique 
conjecture.	 
One may imagine that the complexity of the update time for $4$-cycles should be somewhere between 
$O(\sqrt{m})$ for triangles and $O(m)$ for $4$-cliques.
This raises the following question:
 \begin{quote}
	Question: What is the worst-case update time for counting $4$-cycles in the fully dynamic model? 
\end{quote}

The algorithm of \cite{hanauer_et_al:LIPIcs.SAND.2022.18} is a natural generalization of the approach for 
counting triangles in \cite{DBLP:journals/tods/KaraNNOZ20}, so it seems natural that $O(m^{2/3})$ (which 
sits in the middle of $O(\sqrt{m})$ and $O(m)$) might be the correct answer for the update time for counting 
$4$-cycles.
We show this is \emph{not} the case by giving an algorithm for maintaining the 
number of $4$-cycles in fully dynamic graphs in $O(m^{2/3-\eps})$ worst-case update time for some constant $\eps>0$.
This shows that $O(m^{2/3})$ is not the correct answer for the complexity of the update time. 
Consider the following theorem:
\begin{restatable}{theorem}{AlgGeneral}\label{thm:alg-general}
	There is an algorithm for counting the number of $4$-cycles in any fully dynamic 
	graph in 
	worst-case update time $O(m^{2/3 -\eps})$ 
	where $m$ is the number of edges in the graph at that point and $\eps>0$ is a constant.
	Using the current best value of the square matrix multiplication exponent 
	$\omega=2.371339$
	\iffinaltheorem\footnote{The time complexity of multiplying two $n \times n$ matrices is 
	$O(n^{\omega})$.}\fi 
	we get $\eps=0.009811$ and 
	if we assume the best possible exponent i.e.\ $\omega=2$ then we get 
	$\eps=1/24$.
\end{restatable}

This approach crucially relies on fast matrix multiplication, and we only get an improvement when $\omega<2.5$.
This is very surprising because any upper bound on $\omega$ better than $3$ like Strassen's algorithm \cite{strassen1969gaussian} is not sufficient.
For $\omega=2.371339$ the update time improves from $m^{0.66}$ to $m^{0.65686}$ and if $\omega=2$ it further improves to $m^{0.625}$. 
This improvement to the second digit after the decimal place is significant considering improvements of a similar kind are much smaller like in metric TSP \cite{karlin2021slightly}, random-order streaming matching \cite{assadi_randomstreaming}, and fast matrix multiplication \cite{alman2024more}.
These results, like ours, demonstrate that the previous best bounds were not optimal by providing a slight improvement over them.

\paragraph{Existing Lower Bounds.}
The lower bound of $\Omega(m^{1/2-\gamma})$ for any constant $\gamma>0$ for the update 
time for maintaining the number of triangles in any fully dynamic graph holds under the OMv conjecture, so it 
even holds against algorithms that use fast matrix multiplication.
However, the lower bound of $\Omega(m^{1-\gamma})$ for any constant $\gamma>0$ for the update 
time for maintaining the number of $4$-cliques in any fully dynamic graph holds under the static 
combinatorial $4$-clique conjecture, so it does not hold against algorithms that use fast matrix 
multiplication.	 
Thus, it may be possible to develop a faster dynamic algorithm for $4$-cliques using fast matrix multiplication 
without breaking the lower bound.
We note that our algorithm is non-combinatorial because it uses fast matrix multiplication, so we do not rule 
out a combinatorial lower bound of $\Omega(m^{2/3})$ for the update time for maintaining the number of 
$4$-cycles in a fully dynamic graph.

We will show in \Cref{sec:cycles-general-graphs} that the problem of dynamically maintaining 
the number of $4$-cycles in general graphs is equivalent to dynamically maintaining the number of layered $4$-cycles in $4$-layered graphs (which we define formally in \Cref{subsec:prelim-notation} and informally below).
$4$-layered graphs are graphs with $4$ sets of vertices which we call layers each 
of which are 
independent sets and the edges of the graph only exist between consecutive layers 
(there can also 
be edges between the last and first layer).
These graphs are basically $4$-partite graphs with additional constraints on the edges (that edges exist only between consecutive layers).
A layered $4$-cycle in a $4$-layered graph is a $4$-cycle that has exactly one vertex in each layer and edges connecting these vertices in adjacent layers.

As we discussed earlier, counting layered cycles in layered graphs is interesting from the perspective of database theory.
Suppose we have attributes $L_1,L_2,L_3$ and $L_4$ and binary relations $A(L_1,L_2)$, $B(L_2,L_3)$, 
$C(L_3,L_4)$, and $D(L_4,L_1)$ defined on these attributes.
We can construct a $4$-layered graph $G$ to represent this information, where each layer corresponds to 
one of the attributes $L_1,L_2,L_3,$ and $L_4$.
Each element of the attributes is represented as a vertex in the graph $G$.
Also, each tuple in the relations corresponds to an edge in $G$.
If we compute the join of $A,B,C$ and $D$ we are effectively finding tuples $(a,b,c,d)$ such that: $(a,b)\in 
A$, $(b,c)\in B$, $(c,d)\in C$, and $(d,a)\in D$.
Each resulting tuple $(a,b,c,d)$ from the join corresponds to a unique $4$-cycle in the graph $G$.
Therefore, the size of the join is equal to the number of layered $4$-cycles in $G$.

Thus, the problem of computing the size of a join of four binary relations is equivalent to the problem of 
counting the number of layered $4$-cycles in $4$-layered graphs.
Since maintaining the number of $4$-cycles in a fully dynamic graph is equivalent to maintaining the number 
of layered $4$-cycles in fully dynamic $4$-layered graphs, we solve the latter problem.
We state the result for $4$-layered graphs below:
\begin{restatable}{theorem}{AlgLayer}\label{thm:alg-layer}
There is an algorithm for counting the number of layered $4$-cycles in any fully dynamic 
$4$-layered graph in worst-case update time $O(m^{2/3 -\eps})$ 
where $m$ is the number of edges in the graph at that point and $\eps>0$ is a constant.
Using the current best value of the square matrix multiplication exponent 
$\omega=2.371339$ we get 
$\eps=0.009811$ and 
if we assume the best possible exponent i.e.\ $\omega=2$ then we get 
$\eps=1/24$.
\end{restatable}

This is the first result to demonstrate that fast matrix multiplication can be leveraged to achieve a speedup in join size estimation problems in dynamic graphs. While this approach has been previously explored for static algorithms \cite{tetekicalp22,XiaoFMMQuery}, it had not been extended to the dynamic setting until now.
For example, in the static case, fast matrix multiplication enables solving the triangle query in time $O(m^{\frac{2\omega}{\omega + 1}})$ \cite{alon1997finding}, which is faster than the traditional combinatorial $m^{3/2}$ algorithm since $\omega < 3$. However, in dynamic settings, such a gap between combinatorial and non-combinatorial algorithms had not been observed. The triangle query, for instance, can be maintained in $O(m^{1/2})$ update time using a combinatorial algorithm, which is optimal under the OMv conjecture. 
A widely held belief was that fast matrix multiplication does not help in dynamic settings, as algorithms like Strassen's require prior knowledge of the two input matrices. 
Consequently, it was assumed that Incremental View Maintenance (IVM) would not benefit from fast matrix multiplication.
Our result challenges that intuition by providing a counterexample using the $4$-cycle query. 
This insight could potentially inspire a new class of IVM algorithms.

It is also important to note that we get a worst-case bound on the update-time and not an amortized bound. On the surface it looks like we are multiplying large matrices, so the bound should be amortized, but we answer initial queries quickly because there are only a few edges and in parallel compute the matrix products using fast matrix multiplication to speed up later queries.
We now discuss our techniques in more detail.

\subsection*{Our Techniques}
To solve the problem of maintaining the total number of $4$-cycles we count the 
number of cycles through the new update.
Thus, the algorithm remains the same (counts the number of cycles through the new edge update) whether 
the update is an insertion or deletion.
The answer for the total number of $4$-cycles after the update is the total number of 
$4$-cycles before the update plus or minus the number of cycles through the new edge update depending 
on whether it is an insertion or deletion.
We count the number of cycles through the new edge update by finding the number of $3$-paths 
between the endpoints of the edge update.
In the simple algorithm (\Cref{sec:alg-trivial}), we do this by storing a data structure 
for the number of $2$-paths (paths with $2$ edges also called \emph{wedges}) between all pairs of vertices.
During a query $(u,v)$, we count the number of $3$-paths between $u,v$ by going over all neighbors of $u$ and then using the stored 
number of wedges to $v$ and then summing them up to get the number of $3$-paths.
This gives a worst-case update time of $O(n)$.
We now want to get the update time purely as a function of $m$, the number of edges.
The goal is to minimize the worst-case update time over all edge updates.
This might be slightly tricky because we want the update time to be purely a function of $m$, the current 
number of edges in the graph, but $m$ changes with every insertion or deletion.
Thus, the goal is to come up with an algorithm with a worst-case update time of $O(m^x)$ with the smallest 
possible value of $x$.

\paragraph{Algorithm of Previous Work}
\cite{hanauer_et_al:LIPIcs.SAND.2022.18} is successful in getting the update time purely as a function of $m$ 
and give an upper bound of $O(m^{2/3})$ for the worst-case update time.
Their approach is to group vertices of the graph into high and low degree vertices and compute different data structures for them which can be maintained in $O(m^{2/3})$ update time and use them to answer queries in $O(m^{2/3})$ update time.
The data structures they maintain for low degree vertices are wedges through low degree vertices and 
$3$-paths through two low degree vertices, both of which can be maintained because iterating over pairs of 
neighbors of low degree vertices can be done easily.
For the high degree vertices, they maintain wedges through high degree vertices but only for pairs of vertices 
that have a high degree. Maintaining this is easy because the number of high degree vertices is small.
For the query, the number of paths through two low degree vertices is stored.
The number of paths through one high and one low degree vertex can be found easily by iterating over all high degree vertices and using the wedges through low degree vertices.
To calculate the number of paths through two high degree vertices, they do different things depending on the degrees of the endpoints of the query. 
If the endpoints have a low degree then they iterate over their neighbors otherwise they use the wedges through high degree vertices.
This is enough to give them an update time of $O(m^{2/3})$.

\paragraph{Our Algorithm}
To improve the upper bound of \cite{hanauer_et_al:LIPIcs.SAND.2022.18} we crucially use fast matrix 
multiplication. 
We group the vertices into many classes based on their degree and the layer they are in and create data 
structures for them and use them to answer queries in update time $O(m^{2/3-\eps})$. 
The main technical part is to solve the problem under the assumption that the number of vertices in each layer $n=m^{2/3}$ (this is technically 
incorrect because we assume $n=m^{2/3+2\eps}$ but this is morally correct).
Dropping this assumption is highly non-trivial, but we show how to do that in \Cref{sec:tiny-vertices}.
The idea is that for vertices that have a very small degree, we can answer queries for them and compute 
paths through them in $O(m^{2/3-\eps})$ update time, so we can ignore them for the rest of the analysis.
This gives us a lower bound on the degree of the remaining vertices and hence we can assume that the 
number of relevant vertices is small ($m^{2/3}$).

The vertices in the first and last layer are divided into high, medium, and low degree.
The vertices in the second and third layer are divided into high and low degree called dense and sparse 
vertices respectively.
We can iterate over the neighbors of low and medium degree vertices within the update time, and we can 
iterate over all the medium and high degree vertices within the update time.
The first data structure we store is wedges through sparse vertices (similar to wedges through low degree 
vertices).
Next we store wedges through dense vertices but the endpoints of the wedges have to be either high, 
medium or dense (similar to wedges through high degree vertices, but the endpoints are also high degree).
Using these we can answer queries when one of the endpoints is medium degree which without loss of 
generality is in the first layer.
We can iterate over the neighbors of the vertex in the first layer and then use the number of wedges through 
sparse vertices to the other endpoint.
To get the paths through the dense vertices in the third layer, we iterate over all of them and check if they 
have an edge to the endpoint in the last layer.
We then use the wedges through sparse vertices and wedges through dense vertices to the endpoint in the 
first layer.
These details are similar to the algorithm of \cite{hanauer_et_al:LIPIcs.SAND.2022.18}.

The first difficult case is when both the endpoints of the query are high degree, and we have to calculate the 
number of paths through two sparse vertices.
The algorithm of \cite{hanauer_et_al:LIPIcs.SAND.2022.18} had the number of $3$-paths through two low 
degree vertices stored, but we cannot afford to store the number of $3$-paths through two sparse vertices.
This is where our crucial idea of \textbf{phases} comes in.
We partition the stream of updates into phases of size $m^{1-\delta}$, and it is important that the size of a 
phase is asymptotically smaller than $m$, so we need $\delta>0$.
We then consider eight different types of $3$-paths where each edge of the $3$-path is either in the current 
phase or in an old phase.
For most types of these $3$-paths, we are able to store the number of $3$-paths through two sparse 
vertices like the data structure for the number of $3$-paths through two low degree vertices in 
\cite{hanauer_et_al:LIPIcs.SAND.2022.18}.
This is because the number of high degree vertices when considering the edges of a single phase is 
asymptotically smaller compared to the number of high degree vertices when considering all the edges.
Using these ideas (and some others) we are able to answer queries when both endpoints of the query have a 
high degree.

Another difficult case is when both the endpoints of the query are low degree, and we have to calculate the 
number of paths through two dense vertices.
We cannot iterate over the neighbors of both endpoints like we did in the algorithm of 
\cite{hanauer_et_al:LIPIcs.SAND.2022.18}.
We need to use another property of phases here.
A phase should be long enough so that in the time it takes to process all the edge 
updates in a phase, we are able to multiply two square matrices of dimension $m^{2/3}$.
During a phase, we compute the number of paths between all pairs of vertices in the previous phase using 
fast matrix multiplication.
This allows us to iterate over the neighbors of an endpoint and use the stored number of wedges to the other 
endpoint that go through the edges of the old phases.
To deal with edges that go through the new phase we take advantage of its size and iterate over all pairs of 
dense vertices in the second and third layer (one of which has just $O(m^{1-\delta})$ edges incident on it) 
and 
compute the number of paths through them.
A combination of these ideas along with some others allows us to answer queries when both endpoints have a 
low degree.

%

Note that our algorithm is not as simple as running the algorithm of 
\cite{hanauer_et_al:LIPIcs.SAND.2022.18} for the new phase (this would give a 
better runtime 
because we have asymptotically fewer edges) and then using matrix multiplication to calculate the number of 
paths between all pairs of vertices in the old phases. 
This is because we also need to look at paths that have some edges in the old 
phases and some in 
the new phase.
We spell out the full details of the algorithm in \Cref{sec:main-alg} with the setup 
being in 
\Cref{sec:setup-main}.
We also state a warm up to the main algorithm in \Cref{sec:alg-simple} which solves 
the problem 
under an assumption.
The purpose of stating this algorithm is  two-fold. 
The algorithm is simpler so this gives a flavor for the main algorithm, and we will also 
use it as a 
subroutine in our main algorithm.


\paragraph{Related Work.}
We now discuss relevant related work.
\cite{eppstein2009h} studied the problem of counting the number of $3$-vertex subgraphs in a fully 
dynamic graph and gave an algorithm with update time $O(h)$ where $h$ is the $h$-index of the graph, the 
maximum number such that the graph contains $h$ vertices of degree at least $h$.
\cite{EPPSTEIN201244} extend the results of \cite{eppstein2009h} to directed subgraphs and to subgraphs 
of size $4$.
For the directed case, they provide an algorithm for all three-vertex induced subgraphs in $O(h)$ amortized 
update time. 
For the undirected case, they provide an algorithm for maintaining the counts of $4$-vertex subgraphs in 
$O(h^2)$ amortized update time. 
\cite{henzinger2022complexity} showed that counting $4$-cycles is hard even in random graphs by giving a 
lower bound of $\Omega(m^{1/2-\gamma})$ for any $\gamma>0$, for the update time.
Finally, \cite{hanauer2022recent} provides a comprehensive survey on recent advances in fully dynamic 
graph algorithms.

Outside maintaining the join size in dynamic graphs, there has also been related work on join algorithms \cite{ngo2014skew,NgoOptimalJoin}, join size estimation and join sampling \cite{Dengetal23,Kimetal23} among other areas.
There is also a vast body of work on cycle detection and approximately counting the number of cycles for both directed and undirected graphs starting with \cite{alon1997finding,Yusteretal97} and and culminating in the current best bounds \cite{tetekicalp22,censorhillel_et_al24} (see additional references therein).
The problem of approximately counting subgraphs, particularly triangles and $4$-cycles, has attracted significant attention in the streaming setting \cite{JowhariNew05,BuriolCounting06,BecchettiEfficient10,KaneCounting12,McGregorVV16,KallaugherMPV19,McGregorVorotnikova20,vorotnikova2020improved}. Finally, there is also a long line of work for listing cycles beginning with the classic algorithms \cite{TarjanEnumeration73,JohnsonFinding75} with recent advancements focusing on efficiently listing $6$-cycles in both dense and sparse graphs \cite{Jinlisting24,williams2024listing6cyclessparsegraphs}.


\section{Preliminaries}\label{sec:prelim}

\subsection{Notation.}\label{subsec:prelim-notation}
For a graph $G=(V,E)$, we use $n$ to represent the number of vertices ($\card{V}$) and $m$ to represent the number of edges ($\card{E}$).
We use $\deg(v)$ and $N(v)$ for each vertex $v \in V$ to denote the degree (the number of edges incident on the vertex) and neighborhood (the set of neighbors) of $v$, respectively. 
Note that in this paper, we will only consider unweighted undirected simple graphs (no multi-edges and self loops). 

A \textbf{$k$-path} in a graph is a path with $k$-edges. 
Note that we only consider simple paths, so every vertex on the path has degree exactly $2$ except for the starting and ending vertex which have degree $1$.
Formally, a $k$-path in $G$ is a sequence of distinct vertices $(v_0, v_1, \dots, v_k)$ such that $(v_i, v_{i+1}) \in E$ for all $0 \leq i < k$. 
In this paper, we will mostly just use $2$-paths (also called wedges) and $3$-paths.

A \textbf{$k$-cycle} in a graph is a cycle with $k$-edges (equivalently $k$ vertices).
We only consider simple cycles, so every vertex in the cycle has degree exactly $2$.
Formally, a $k$-cycle in $G$ is a sequence of vertices $(v_0, v_1, \dots, v_{k-1}, v_0)$ such that $(v_i, v_{i+1}) \in E$ for all $0 \leq i < k$ (with indices taken modulo $k$), and where all vertices $v_0, v_1, \dots, v_{k-1}$ are distinct. 
In this paper, we will primarily focus on $4$-cycles.

A \textbf{$4$-layered graph} is a graph $G=(V,E)$ where $V= L_1 \cup L_2 \cup L_3 \cup L_4$ is a union of four layers each containing $n$ vertices.
Each layer is an independent set (no edges exist between vertices in the same layer).
Additionally, edges only exist between the four consecutive layers and are represented by the matrices $A, B, C,$ and $D$. 

A \textbf{layered $4$-cycle} in a $4$-layered graph $G=((L_1,L_2,L_3,L_4),E)$ is a $4$-cycle defined by a sequence of vertices $(v_1, v_2, v_3, v_4, v_1)$ such that $v_i \in L_i$ and $(v_i, v_{i+1}) \in E$ for all $i \in [4]$ (where $v_5=v_1$).

We can similarly define layered $2$-paths and layered $3$-paths which have all of its vertices in different layers.
Throughout the paper, we are solely focused on layered cycles and paths in $4$-layered graphs, so
for convenience, we will drop the word ``layered''.

\begin{figure}[t]
	\centering
	\small
	\begin{minipage}{0.35\textwidth}
		\centering
	\centering
	\begin{tabular}{ |p{2cm}|p{2cm}|  }
		\hline
		\multicolumn{2}{|c|}{\textbf{Databases}} \\
		\hline
		\hline	
		{A ($L_1,L_2$)}& {B ($L_2,L_3$)} \\
		\hline
		(1,1) &(1,1)\\
		\hline
		(1,2) &(2,1)\\
		\hline
		(1,3) &(3,1)\\
		\hline
		(2,2) &(3,3)\\
		\hline
		(3,2) &\\
		\hline
	\end{tabular}
	\end{minipage}
	\hfill
	\begin{minipage}{0.2\textwidth}
		\centering
	\centering
	\begin{tabular}{ |p{2cm}|  }
		\hline
		\multicolumn{1}{|c|}{\textbf{$A \Join B$}} \\
		\hline
		(1,1,1)\\
		\hline
		(1,2,1)\\
		\hline
		(1,3,1) \\
		\hline
		(1,3,3) \\
		\hline
		(2,2,1)\\
		\hline
		(3,2,1) \\
		\hline
	\end{tabular}
	\end{minipage}
	\hfill
	\begin{minipage}{0.35 \textwidth}
		\centering
    		\resizebox{110pt}{80pt}{ \begin{tikzpicture}
	\node[vertexA] at (0,0) (a1) {};
	\node[vertexA] at (0,1) (a2) {};
	\node[vertexA] at (0,2) (a3) {};
	\node[vertexA] at (2,0) (b1) {};
	\node[vertexA] at (2,1) (b2) {};
	\node[vertexA] at (2,2) (b3) {};
	
	\node[vertexA] at (4,0) (c1) {};
	\node[vertexA] at (4,1) (c2) {};
	\node[vertexA] at (4,2) (c3) {};
			
	\draw[edge] (a1) -- (b3);
	\draw[edge] (a2) -- (b2);
	\draw[edge] (a1) -- (b1);
	\draw[edge] (a1) -- (b2);
	\draw[edge] (a3) -- (b2);
	
	\draw[edge] (b3) -- (c1);
	\draw[edge] (b3) -- (c3);
	\draw[edge] (b2) -- (c1);
	\draw[edge] (b1) -- (c1);
	
	\draw (1,-0.5) node [anchor=north west][inner sep=0.75pt]   [align=left] 
	{{$A$}};
	
	\draw (3,-0.5) node [anchor=north west][inner sep=0.75pt]   [align=left] 
	{{$B$}};

\draw (0,3) node [anchor=north west][inner sep=0.75pt]   [align=left] 
{{$L_1$}};

\draw (2,3) node [anchor=north west][inner sep=0.75pt]   [align=left] 
{{$L_2$}};

\draw (4,3) node [anchor=north west][inner sep=0.75pt]   [align=left] 
{{$L_3$}};

\draw (-0.7,2.2) node [anchor=north west][inner sep=0.75pt]   [align=left] 
{{$3$}};

\draw (-0.7,1.2) node [anchor=north west][inner sep=0.75pt]   [align=left] 
{{$2$}};

\draw (-0.7,0.2) node [anchor=north west][inner sep=0.75pt]   [align=left] 
{{$1$}};

\end{tikzpicture} }			
	\end{minipage}
	
	\caption{Binary Relations $A,B$. Every element of the join ($A \Join B$) corresponds to a $2$ -paths in the layered graph between $L_1$ and $L_3$. Thus, the size of the join ($A \Join B$) is the number of $2$-paths in the layered graph.}
	\label{fig:join-layered}
\end{figure}
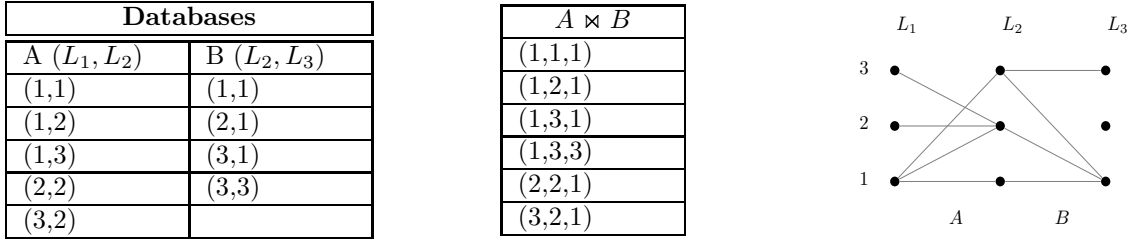

\paragraph{Fast Matrix Multiplication.}
We now discuss some notation for fast matrix multiplication.
Note that $\omega$ is the exponent of square matrix multiplication. 
This means that the time it takes 
to multiply two $n \times n$ matrices is $O(n^{\omega})$.
The current best known upper bound on $\omega$ is $2.371339$ 
\cite{alman2024more}.
We will also need fast rectangular matrix multiplication in our algorithm.
We define $\omega(a,b,c)$ to be the exponent of rectangular matrix multiplication 
of an $n^a \times 
n^b$ dimension matrix by an $n^b \times n^c$ dimension matrix.
\cite{alman2024more} also gives improved upper bounds for $\omega(a,b,c)$.

\subsection{Equivalent Queries.}
After every update to the simple graph, the query we have to answer is the exact total number (count) of distinct $4$-cycles in the current graph.
We will show in \Cref{sec:cycles-general-graphs} that this problem is equivalent in general and layered graphs, so we will just focus on layered graphs.
Thus, the new problem is the following.
After every update to the simple $4$-layered graph, the query we have to answer is the exact total number of distinct layered $4$-cycles in the current graph.
This is the same problem as having $4$ binary relations that are continuously updated, where the query answered after every update is the exact total number of tuples in their cyclic join.

It is easy to see that maintaining the total number of layered $4$-cycles is equivalent to just counting the number of new layered $4$-cycles i.e.\ cycles that go through the new edge update (insertion or deletion).
This is because we can add or subtract cycles that go through the new edge update from the previous answer to get the new answer.
We answer the query of the number of new layered $4$-cycles by counting the number of $3$-paths between the endpoints of the edge update.

Without loss of generality, we assume that queries have to be answered only when there is an edge 
update in $D$ because we can run $4$ copies of this algorithm, one for queries in $D$ and the others for queries in $A,B,$ and $C$.
After each update, we ask the relevant instance of the algorithm to provide the answer.
This means we now focus on queries in $D$ while handling edge updates in $A,B,$ and $C$. 
To efficiently answer queries in $D$, we maintain appropriate data structures for edges in $A,B$ and $C$.
Note that we maintain the same data structures for all $4$ copies of the algorithm (corresponding to matrices $A,B,C$ and $D$). 
By symmetry, the analysis is the same for all of them, so we just focus on the case when queries are in $D$ and edge updates are everywhere else.
The final equivalent problem we will solve is the following:
 \begin{quote}
We are given a $4$-layered graph $G$ which undergoes edge updates in $A,B$ and $C$ each of which must be processed in $O(m^{2/3-\eps})$ worst-case time. 
At any point, a query $(u,v)$ can be made ($u \in L_1$ and $v \in L_4$) for the number of $3$-paths between $u$ and $v$ that go through $A,B$ and $C$ which should be answered in $O(m^{2/3-\eps})$ worst-case time.
\end{quote}

\section{Warm-up: When A and C are Fixed}\label{sec:alg-simple}
In this section, we give an algorithm for counting $4$-cycles in a $4$-layered graph under a few 
assumptions.

\subsection{Setup}\label{sec:setup-simple}

We let the worst-case update time be $O(m^{2/3-\e1})$ and we will show we can get $\e1>0$ to 
be a constant.
We first discuss the data structures we need to store and then show how the queries can be 
answered in time $O(m^{2/3-\e1})$.

Before we discuss the data structures, we group the vertices into different \textbf{classes}.
The vertices in $L_1$ (resp. $L_4$) are partitioned into $H,M,L$ based on their degree in $A$ (resp. 
$C$).
\begin{itemize}
	\item \textbf{High} ($H$): degree between $m^{2/3-\e1}$ and $n$.
	\item \textbf{Medium} ($M$): degree between $m^{1/3+\e1}$ and $2 m^{2/3-\e1}$.
	\item \textbf{Low} ($L$): degree between $0$ and $2 m^{1/3+\e1}$.	
\end{itemize}

The edge updates to $B$ are divided into \textbf{chunks} of size $m^{2/3 - \e1}$ called 
$B_1,B_2,\ldots B_i$.
For a chunk $B_i$, vertices in $L_2$ and $L_3$ with degree in $B_i$ at most $m^{1/3-\e2}$ are 
in \textbf{Sparse} ($S$) and degree in $B_i$ at least $m^{1/3-\e2}$ are in 
\textbf{Dense} ($D$).
Note that these classes are different for each $B_i$, so a vertex could be in $S$ for 
some chunk $B_i$ and in $D$ for some other chunk $B_j$.
We also partition the edges of $B_i$ based on whether the endpoints are in $D$ or 
$S$. Thus, 
$B_i$ can be written as a sum of four matrices: $B_i= B_{i,DD} + B_{i,SS} + B_{i,DS} 
+ B_{i,SD}$.
Also, we denote by $B_{<i}:=\sum_{j=1}^{i-1} B_j$. Similarly, we have $B_{<i 
,DD}:=\sum_{j=1}^{i-1} 
B_{j,DD}$ for all the four subsets of $B_i$.
Note that these matrices may contain numerous rows and columns consisting entirely of zeros, effectively reducing their dimension for computational purposes.

We use superscripts to talk about a submatrix of a matrix by restricting the vertices in different 
layers. 
We also use $*$ to represent all the vertices in a layer.
For instance $A^{H *}$ represents the submatrix of $A$ where we only look at the 
edges 
of $A$ between $L_1^H$ and $L_2$. 
We use a subscript to represent a subset of the edges.
For instance $B_i$ contains only the edges in chunk $i$ and $B_{i,DD}$ only 
contains the edges 
between dense vertices of $B_i$.
To make things less cumbersome we will drop some parts of the notation when things are clear 
from context.
We now state the assumptions we need in this section:
\begin{restatable}{assumption}{AsmVertices}\label{asm:vertices}
	The number of vertices in each layer is $n \leq m^{2/3+2\eps}$.
\end{restatable}

\begin{restatable}{assumption}{AsmClasses}\label{asm:classes}
	The vertices do not change classes throughout the algorithm.
\end{restatable}

\begin{assumption}\label{asm:AC-fixed}
	There are no edge updates in $A$ and $C$ i.e.\ the only edge updates are in $B,D$.
\end{assumption}


We will prove the following lemma in this section:
\begin{lemma}\label{lem:layered-AC-fixed}
	\Cref{thm:alg-layer} holds under \Cref{asm:vertices}, \Cref{asm:classes}, and 
	\Cref{asm:AC-fixed}.
\end{lemma}

\paragraph{High Level Idea.}
We would like to store wedges between all pairs of vertices like in the simple 
algorithm in \Cref{sec:alg-trivial} and then during a query, iterate over the all the neighbors of one 
of the query 
vertices and use the number of wedges to the other query vertex.
There are two major problems with this approach.
We do not have enough time to iterate over all the neighbors of a high degree 
vertex. 
Also, iterating over all low degree vertices in a layer takes too much time, and so we 
cannot store the 
wedges for low degree vertices.

Thus, we create different data structures for vertices in $H,M,L$ which we show in 
the next 
subsection (\Cref{subsec:DS-simple}). 
We directly store the number of $3$-paths between pairs of high degree vertices 
because we 
cannot iterate over their neighbors.
We store the wedges for high and medium degree vertices and some types of 
wedges for the low 
degree vertices.
Lastly, we show how we use these data structures to answer queries in the 
subsection after we 
show the data structures (\Cref{subsec:queries-simple}).

\subsection{Data Structures}\label{subsec:DS-simple}

\renewcommand{\arraystretch}{1.5} 


\begin{table}[h!]
	\centering
\begin{tabular}{ |p{4cm}|p{4cm}|p{4cm}|  }
	\hline
	\multicolumn{3}{|c|}{\textbf{Data Structures}} \\
	\hline
	\hline	
	\textbf{High}& \textbf{Medium} &\textbf{Low} \\
	\hline
	$A^{H*} \cdot B_{<i}$ (\Cref{eq:DS-H-simple}) & $A^{M*} \cdot B_{<i}$ (\Cref{eq:DS-M-simple}) & 
	$A^{L*} \cdot B_{<i,DD}$ (\Cref{eq:DS-L-simple})\\
	\hline
	$A^{H*}\cdot B_{<i} \cdot C^{*H} $ (\Cref{eq:DS-H-simple}) & $ B_{<i} \cdot C^{*M}$ 
	(\Cref{eq:DS-M-simple}) & $ A^{L*} \cdot 
	B_{<i,S*}$ (\Cref{eq:DS-L-simple}) \\
	\hline
	$B_{<i} \cdot C^{*H}$	(\Cref{eq:DS-H-simple}) & & $B_{<i,DD} \cdot C^{*L}$ (\Cref{eq:DS-L-simple}) 
	\\
	\hline
	& & $B_{<i,*S} \cdot C^{*L}$ (\Cref{eq:DS-L-simple}) \\
\hline
\end{tabular}
\caption{Data Structures used in the Algorithm where $A$ and $C$ are fixed.}
\label{tab:DS-simple}
\end{table}

We first show what data structures we maintain.
We maintain some data structures on the fly as the updates arrive and show that the update time is 
$O(m^{2/3-\e1})$.
For the remaining data structures, we compute them for the previous chunk when the updates of the next chunk arrive. 
We show that this process takes $O(m^{4/3-2\e1})$ total time, so the computation can be spread over the $m^{2/3-\e1}$ updates that arrive during the chunk, leading to our desired worst-case update time.
The list of data structures can be found in \Cref{tab:DS-simple}.

\begin{lemma}\label{lem:DS-update-time-simple}
	The worst-case update time for all the data structures we store is $O(m^{2/3-\e1})$ during an 
	edge update and $O(m^{4/3-2\e1})$ if it is being updated during the insertion of a chunk.
\end{lemma}

Consider one of these data structures, for instance $A^{H*} B_{<i}$. It can be expressed in database terminology as follows:
We start with the binary relation $A$ and select the tuples where the first attribute corresponds to a high-degree vertex in $L_1$, denoted as $A^{H*}$. For the binary relation $B$, we filter the tuples to include only those belonging to the first $i-1$ chunks, denoted as $B_{<i}$. We then perform a join between the filtered relations $A^{H*}$ and $B_{<i}$ on their shared attribute (vertices in $L_2$). 
Finally, we aggregate the result over pairs of vertices in $L_1$ and $L_3$, by computing the count of two-hop paths between them.

\paragraph{Data Structures for High vertices.}
We maintain the following data structures for high degree vertices: 
\begin{align}\label{eq:DS-H-simple}
	A^{H*} \cdot B_{<i} \; , \;\; A^{H*}\cdot B_{<i} \cdot C^{*H} \text{ and } \; B_{<i} \cdot 
C^{*H}.	
\end{align}

\begin{claim}\label{clm:H-online-simple}
	We can maintain $A^{H*} \cdot B_i$ and $B_i \cdot C^{*H}$ on the fly in 
	worst-case update time 
	$O(m^{1/3+\e1})$.
\end{claim}
\begin{proof}
	Consider the data structure $A^{H*} \cdot B_i$ (the analysis for $B_i \cdot 
	C^{*H}$ is symmetric).
	When we get an edge update $(u,v)$ in $B_i$ we iterate over all $m^{1/3+\e1}$ vertices in 
	$L_1^H$ and update the data structure.
\end{proof}

\begin{claim}\label{clm:H-offline-simple}
	Multiplying $(A^{H*} \cdot B_i)$ with $C^{*H}$ can be done in time 
	$m^{\omega({1/3+\e1},{2/3-\e1},{1/3+\e1})}$.
\end{claim}
\begin{proof}
	We know that $B_i$ has at most $m^{2/3-\e1}$ edges.
	Thus, the dimension of $(A^{H*} \cdot B_i)$ is $m^{1/3+\e1} \times m^{2/3-\e1}$ 
	and the 
	dimension of $C^{*H}$ is $m^{2/3-\e1} \times m^{1/3+\e1}$.
	This is because we only have to consider vertices in $L_3$ that have a non-zero degree, giving us a dimension of $m^{2/3-\e1}$ instead of $m^{2/3+2\eps}$.
	Thus, using rectangular fast matrix multiplication we can multiply $(A^{H*} \cdot 
	B_i)$ with $C^{*H}$ 
	in time $m^{\omega({1/3+\e1},{2/3-\e1},{1/3+\e1})}$.
\end{proof}

We need to compute this data structure for every chunk in the time that the next chunk arrives.
This gives us a constraint $m^{\omega({1/3+\e1},{2/3-\e1},{1/3+\e1})}\leq m^{4/3-2\e1}$ 
which implies:
\begin{constraint}\label{constraint:DS-H-simple}		
		{\omega({1/3+\e1},{2/3-\e1},{1/3+\e1})} \leq 4/3 -2\e1.
\end{constraint}

\paragraph{Data Structures for Medium vertices.}
We maintain the following data structures for medium degree vertices: 
\begin{equation}\label{eq:DS-M-simple}
	A^{M*} \cdot B_{<i} \text{ and } \; B_{<i} \cdot C^{*M}.	
\end{equation}

\begin{claim}\label{clm:M-online-simple}
	We can maintain $A^{M*} \cdot B_i$ and $B_i \cdot C^{*M}$ in time 
	$O(m^{2/3-\e1})$.
\end{claim}
\begin{proof}
	Consider the data structure $A^{M*} \cdot B_i$ (the analysis for $B_i \cdot 
	C^{*M}$ is symmetric).
	When we get an edge update $(u,v)$ in $B_i$ we iterate over all $m^{2/3-\e1}$ vertices in 
	$L_1^M$ and update the data structure.
\end{proof}

\paragraph{Data Structures for Low vertices.}
Consider the low degree vertices in $L_1$ (analogous for $L_4$).
$A^{L*}$ is the submatrix of $A$ only incident on low degree vertices in $L_1$.
We compute the following data structures: 
\begin{equation}\label{eq:DS-L-simple}
	A^{L*} \cdot B_{<i,DD} \; , \;\;  A^{L*} \cdot B_{<i,SS} \; , \;\; A^{L*} \cdot
B_{<i,SD} \text{ and } \; B_{<i,DD} \cdot C^{*L} , \;\; B_{<i,SS} \cdot C^{*L} , \;\; B_{<i,DS} \cdot 
C^{*L}.
\end{equation}

\begin{claim}\label{clm:L-HH-simple}
	Computing $A^{L*} \cdot B_{i,DD}$ can be done in time 
	$m^{\omega({2/3+2\eps},{1/3-\e1+\e2},{1/3-\e1+\e2})}$.
\end{claim}
\begin{proof}
	$B_i$ has at most $m^{1/3-\e1+\e2}$ dense vertices in $L_2$ and $L_3$ and the 
	total 
	number of vertices in $L_1$ is at most $m^{2/3+2\eps}$.
	We need to multiply the matrix $A^{L*}$ of dimension $m^{2/3+2\eps} \times 
	m^{1/3-\e1+\e2}$ 
	with the matrix $B_{i,DD}$ of dimension $m^{1/3-\e1+\e2} \times 
	m^{1/3-\e1+\e2}$.
	Note that the dimension of $A^{L*}$ is $m^{2/3+2\eps} \times m^{2/3+2\eps}$, 
	but we restrict 
	ourselves to the non-zero degree vertices in $L_2$ when considering the edges 
	of $B_{i,DD}$.
	We use rectangular fast matrix multiplication for this.
	Thus, the time taken for fast matrix multiplication is 
	$m^{\omega({2/3+2\eps},{1/3-\e1+\e2},{1/3-\e1+\e2})}$.
\end{proof}

This gives us a constraint $m^{\omega({2/3+2\eps},{1/3-\e1+\e2},{1/3-\e1+\e2})} \leq m^{4/3-2\e1}$ 
which implies:
\begin{constraint}
\label{constraint:DS-L-HH-simple}		
		\omega({2/3+2\eps},{1/3-\e1+\e2},{1/3-\e1+\e2}) \leq 4/3 -2\e1.
\end{constraint}

\begin{claim}\label{clm:L-L-simple}
	Computing $A^{L*} \cdot B_{i,SS}$ and $A^{L*} \cdot B_{i,SD}$ can be done in 
	time $m^{4/3 
	+\e1 
	- \e2}$.
\end{claim}
\begin{proof}
	Go over all $m^{2/3+2\eps}$ low vertices in $L_1$ and then go over their neighbors in 
	$L_2$ that are sparse in $B_i$ and then go over their neighbors in $L_3$ and 
	update the matrix 
	product.
	The degree of the low vertices in $L_1$ is at most $m^{1/3+\e1}$ and the sparse 
	vertices of 
	$B_{i,L*}$ in $L_2$ have degree at most $m^{1/3-\e2}$ in $B_i$.
	Thus, the total time taken is $m^{4/3 +\e1 - \e2 +2\eps}$.
\end{proof}
This gives us a constraint $m^{4/3 +\e1 - \e2 +2\eps} \leq m^{4/3-2\e1}$ which implies:
\begin{constraint}
\label{constraint:DS-L-L-simple}		
		3\e1 +2\eps \leq \e2.
\end{constraint}

\begin{proof}[Proof of \Cref{lem:DS-update-time-simple}]
\Cref{clm:H-online-simple,clm:H-offline-simple,clm:M-online-simple,clm:L-HH-simple,clm:L-L-simple}
together prove the statement in \Cref{lem:DS-update-time-simple} for $B_i$.
We now want to extend this to hold for $B_{<i}$.	

We compute these data structures for each new chunk $B_j$.
Therefore, by induction, when computing the data structure for $B_{i-1}$ we 
already have the data 
structure for $B_{<i-1} = \sum_{j=1}^{i-2} B_j$. When we are done computing the 
data structure for 
$B_{i-1}$ we add it to the one of $B_{<i-1}$ and get the data structure for $B_{<i} = 
\sum_{j=1}^{i-1} 
B_j$.
Observe that for this to work we need $A$ and $C$ to remain the same. 
\end{proof}

\subsection{Queries:}\label{subsec:queries-simple}

\begin{figure}[t]
	\centering 
    \resizebox{150pt}{90pt}{ \begin{tikzpicture}
	\node[vertexBlack] (a1) at (0,0)  {};
	\node[vertexBlack] at (0,1) (a2) {};
	\node[vertexBlack] at (0,2) (a3) {};
	\node[vertexRed] at (0,3) (a4) {};
 
	\node[vertexBlack] at (2,0) (b1) {};
	\node[vertexBlack] at (2,1) (b2) {};
	\node[vertexBlack] at (2,2) (b3) {};
    \node[vertexBlack] at (2,3) (b4) {};

	\node[vertexBlack] at (4,0) (c1) {};
	\node[vertexBlack] at (4,1) (c2) {};
	\node[vertexBlack] at (4,2) (c3) {};
	\node[vertexBlack] at (4,3) (c4) {};

	\node[vertexBlack] at (6,0) (d1) {};
	\node[vertexBlack] at (6,1) (d2) {};
	\node[vertexBlack] at (6,2) (d3) {};
	\node[vertexRed] at (6,3) (d4) {};

	\draw[edge] (c2) -- (b4);
	\draw[edge] (c3) -- (b3);
	\draw[edge,blue] (c2) -- (b2);
	\draw[edge,blue] (c2) -- (b3);
	\draw[edge,blue] (c4) -- (b3);
	\draw[edge,bend left=40,dashed] (a4) to (d4);

    \draw[edge,red] (d4) -- (c4);
    \draw[edge,red] (d4) -- (c2);
    \draw[edge,red] (d4) -- (c1);

    \draw[edge,blue] (a4) -- (b3);
    \draw[edge,blue] (a4) -- (b2);

\draw (1,-0.5) node [anchor=north west][inner sep=0.75pt]   [align=left] 
{\normalsize{$A$}};
	
\draw (3,-0.5) node [anchor=north west][inner sep=0.75pt]   [align=left] 
{\normalsize{$B$}};

\draw (5,-0.5) node [anchor=north west][inner sep=0.75pt]   [align=left] 
{\normalsize{$C$}};

\draw (2.9,0.7) node [anchor=north west][inner sep=0.75pt]   [align=left] 
{\normalsize{$B_{<i}$}};

\draw (-1.3,3) node [anchor=north west][inner sep=0.75pt]   [align=left] 
{\normalsize{$u \in H$}};

\draw (6.3,3.2) node [anchor=north west][inner sep=0.75pt]   [align=left] 
{\normalsize{$v\in M$}};

\draw (-0.2,3.8) node [anchor=north west][inner sep=0.75pt]   [align=left] 
{\normalsize{$L_1$}};
\draw (1.8,3.8) node [anchor=north west][inner sep=0.75pt]   [align=left] 
{\normalsize{$L_2$}};
\draw (3.8,3.8) node [anchor=north west][inner sep=0.75pt]   [align=left] 
{\normalsize{$L_3$}};
\draw (5.8,3.8) node [anchor=north west][inner sep=0.75pt]   [align=left] 
{\normalsize{$L_4$}};
 
\end{tikzpicture} }
	\caption{This figure shows a query $(u,v)$ where $u\in H$ and $v\in M$. We iterate over the neighbors of 
	$v$ and then use the wedges to $u$ stored in the data structure $A^{H*} \cdot B_{<i}$. The neighbors of 
	$v$ are shown using \textcolor{red}{red} edges, the edges in $B_{<i}$ are in gray and the wedges to $u$ 
	are shown in \textcolor{blue}{blue}.}
	\label{fig:HM-query-simple}
\end{figure}
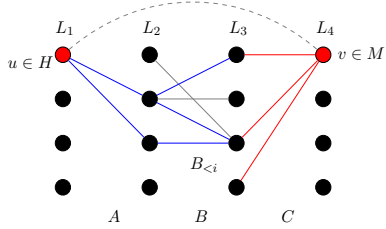	

We have assumed that queries only come in layer $D$.
Let the current query be $e=(u,v)$ where $u$ is in $L_1$ and $v$ is in $L_4$.
Let $B_{i+1}$ be the current chunk in $B$ that was being inserted when we got query $e$.  
Note that the chunk $B_{i+1}$ might be incomplete.

We first describe the procedure of \emph{lazy evaluation} that we use for chunks $B_i$ and $B_{i+1}$.
Given the query $u,v$ we go over all edges of $B_i$ one by one and check if it forms a path from 
$u$ to $v$ (similar for $B_{i+1}$). This takes $O(1)$ time per edge in $B_i$. Thus, we know the 
number of paths from $u$ to $v$ passing through $B_i$ and $B_{i+1}$.
We now have to only look at paths that go through $B_{<i}$.
We will do that in different cases below.

\begin{lemma}\label{lem:query_time-simple}
	All types of queries can be answered in worst-case time $O(m^{2/3-\e1})$.
\end{lemma}
\begin{proof}
We will divide all types of queries into three cases.

\textbf{Case 1: HH.}
In this case, both $u,v$ are high degree vertices.

We know the number of paths from $u$ to $v$ that go through $A, B_{<i}$ and $C$ using 
the data structure $A^{H*} \cdot B_{<i} \cdot C^{*H}$ (\Cref{eq:DS-H-simple}).
The total time taken is $O(1)$.

\textbf{Case 2: HM, ML, HL, MM.}
In this case, $u \in L_1$ is a high or medium degree vertex and $v \in L_4$ is a medium or low 
degree vertex. Note that this also solves the symmetric cases.

To calculate the number of paths from $u$ to $v$ that go through $B_{<i}$ we 
can go over the $2m^{2/3-\e1}$ neighbors of $v$ in $L_3$ and use the data 
structure $A^{H*} \cdot
B_{<i}$ ($A^{M*} \cdot B_{<i}$) maintained in \Cref{eq:DS-H-simple} 
(\Cref{eq:DS-M-simple}) to 
get the number of paths to $u$.
The total time taken is $O(m^{2/3-\e1})$.
\Cref{fig:HM-query-simple} gives an illustration of an $HM$ query.

\textbf{Case 3: LL.}
In this case, both $u,v$ are low degree vertices.

The low degree vertices can be tricky because we do not have $A^{L*} \cdot 
B_{<i,DS}$ stored.
We iterate over the $2m^{1/3+\e1}$ neighbors of $v$ and then compute the number 
of $DD,SS,SD$ 
paths using the data structures $A^{L*} \cdot B_{<i,DD},$ $A^{L*} \cdot 
B_{<i,SS},$ $A^{L*} \cdot 
B_{<i,SD}$ (\Cref{eq:DS-L-simple}).
For the number of $DS$ paths we calculate the paths from the other side in the 
following way.
We iterate over the $2m^{1/3+\e1}$ neighbors of $u$ and use the data 
structure $B_{i,DS} \cdot C^{*L}$ (\Cref{eq:DS-L-simple}).
The total time taken is $O(m^{1/3+\e1})$.
\end{proof}

We make a brief remark here that will be addressed now and will not be mentioned for the remainder of the paper. 
Consider a situation where an edge $e$ is inserted in chunk $B_1$ and subsequently deleted in chunk $B_2$. At first glance, this may seem like a problem when creating data structures for chunk $B_2$, and we may have to go back and change the data structure for chunk $B_1$ incurring a huge cost.
However, this situation can be more effectively understood by treating $e$ as a ``negative edge'' in chunk $B_2$. This approach eliminates the need to alter the data structure for $B_1$ entirely. 
Once we add the data-structures of $B_1$ and $B_2$ together, the entire contribution of edge $e$ will disappear.

For example, consider the data structure $A^{M*} B_i$ and let $x$ be a vertex in $L_1^M$ and $y$ be a vertex in $L_3$.
Let $x$ and $y$ both have edges in $B_1$ to vertices $w$ and $z$ in $L_2$.
Thus, the data structure $A^{M*} B_1$ will have a value of $2$ in position $(x,y)$. Now assume the edge $(w,y)$ is deleted in $B_2$. Then the data structure $A^{M*} B_2$ will have a value of $-1$ in position $(x,y)$. Summing up $A^{M*} B_1$ and $A^{M*} B_2$, we get a value of $1$ indicating one path between $x$ and $y$ in $A^{M*} B_{\leq2}$ which is the correct answer.

Also, note that the computation required for an insertion or deletion is the same. 
One concrete way to take care of this problem is to have $8$ copies of every data structure because edges inside matrices $A,B$ and $C$ either positive or negative. 
Another approach is to maintain a single copy, updating it by adding and subtracting the appropriate values, which may sometimes result in negative values.
In any case, the time complexity increases by at most an $O(1)$ factor.

\subsection{Constraints:}
\label{subsec:constraints-simple}
The simplest constraints come from the fact that the thresholds for the classes should be in 
increasing order. 
So we have $1/3+\e1 \leq 2/3-\e1$ and $1/3-\e2 \leq 2/3-\e1$ implying the following constraints:

\begin{constraint}
\label{constraint:L_one-simple}		
		\e1 \leq 1/6.
\end{constraint}

\begin{constraint}
\label{constraint:L_two-simple}		
		\e1-\e2 \leq 1/3.
\end{constraint}

Solving all these constraints 
(\Cref{constraint:DS-H-simple,constraint:DS-L-HH-simple,constraint:DS-L-L-simple,constraint:L_one-simple,constraint:L_two-simple})
 gives us $\e1 =0.04201965$ and 
$\e2=0.14568075$ when 
$\eps=0.0098109$.
Thus, we get $\e1 \geq \eps$.
We need this because we use this algorithm as a subroutine in our main 
algorithm which has update time $O(m^{2/3-\eps})$.
\Cref{lem:DS-update-time-simple} and \Cref{lem:query_time-simple} along with $\e1 \geq \eps$ 
prove \Cref{lem:layered-AC-fixed}.
Solving all these constraints for the best possible value of $\omega(a,b,c)$
gives us $\e1 = 1/24= 0.0416666$ and 
$\e2=5/24=0.2083333$ when $\eps=1/24$.
The best possible value of $\omega(a,b,c)$ is $\max(a+b,b+c)$ which means that the time it takes to multiply 
the matrices is 
asymptotically the same as the time it takes to read the input. 
\Cref{sec:apx-verify-constraints} shows that these values of $\e1$ and $\e2$, given the value of $\eps$, 
satisfy all the constraints.


\section{Setup for the Main Algorithm}\label{sec:setup-main}
In this section, we setup the problem we want to solve. 
We will eventually prove the following theorem:
\AlgLayer*

The setup will be very similar to the warm-up algorithm in \Cref{sec:alg-simple}.
We let the worst-case update time be $O(m^{2/3-\eps})$ and we will show we can get $\eps>0$ to be a 
constant.
We will discuss the data structures we need to store and then show how the queries can be 
answered using them in time $O(m^{2/3-\eps})$.
Before we discuss the data structures we need, we group the vertices into different classes.


The vertices in $L_1$ and $L_4$ are grouped into classes $L,M,H$ based on their 
degree in $A$ and $C$ respectively.
\begin{itemize}
	\item \textbf{High} ($H$): degree between $m^{2/3-\eps}$ and $n$.
	\item \textbf{Medium} ($M$): degree between $m^{1/3+\eps}$ and $2m^{2/3-\eps}$.
	\item \textbf{Low} ($L$): degree between $0$ and $2m^{1/3+\eps}$.	
\end{itemize}

We group vertices in $L_2$ (resp. $L_3$) into classes $S,D$ based on their combined 
degree in $A,B$ (resp. $B,C$).
\begin{itemize}
	\item \textbf{Dense} ($D$): degree between $m^{2/3-\eps}$ and $n$.
	\item \textbf{Sparse} ($S$): degree between $0$ and $2m^{2/3-\eps}$.
\end{itemize}
The classes for vertices in $L_2,L_3$ have same names here and in \Cref{sec:setup-simple} but 
their definitions are very different. 
The difference is that in the warm-up algorithm we are looking at the degree inside a chunk $B_i$ 
(there are no chunks here), and here we are looking at the combined degree in $A,B$ (or $B,C$).
Also, note that this algorithm does not have chunks, but it will use the algorithm in 
\Cref{sec:alg-simple} as a subroutine.

Notice that these ranges for different classes (in all layers) overlap with each other, so when a vertex 
is in an overlapping region it belongs to both classes.
The overlapping region will be helpful for vertices that transition from one class to the other.


Finally, we talk about the parameters we will use.
The update time of the algorithm will be proved to be $O(m^{2/3-\eps})$ where $\eps$ is a 
parameter.
The goal is to show that $\eps>0$ is a constant.
The algorithm will progress in phases of $m^{1-\delta}$ edge updates where $\delta>0$ is a 
constant. 
Lastly, $\omega$ is the exponent of square matrix multiplication.
These parameters will satisfy the following constraints.

The first constraint we want comes from being able to multiply two square matrices of dimension 
$m^{2/3+2\eps}$ in a phase. The matrices we are interested in multiplying together are $A,B,C$ each of 
which have dimension $m^{2/3+2\eps}$. The reason we do this is so that we can get all types of paths 
between a pair of vertices which belong to the old phases i.e. phases for which this matrix multiplication 
has been done. 
The update time for an entire phase is $m^{1-\delta} \cdot m^{2/3-\eps}$ so the constraint we get is 
$m^{1-\delta} \cdot m^{2/3-\eps} \geq m^{\omega \cdot (2/3+2\eps)}$ which implies:
\begin{constraint}\label{eq:phase-constraint}		
	1-\delta \geq (2\omega +1) \cdot \eps + (\omega-1) 2/3.
\end{constraint}
This constraint will be used in \Cref{subsec:phases}.

The second constraint we want comes from wanting to iterate over pairs of high (or dense) vertices one of 
which is in a new phase and the other one can be in any of the old phases.
The number of high degree vertices in the new phase is $2m^{1-\delta}/m^{2/3-\eps}$ and the number of 
high degree vertices in any of the old phases can be at most $m/m^{2/3-\eps}$.
The product of the two should be at most the update time $O(m^{2/3-\eps})$ so the constraint we get is 
$m^{1/3+\eps} \cdot m^{1-\delta - 2/3+\eps} \leq m^{2/3-\eps}$ which implies:
\begin{constraint}\label{eq:iteration-pairs-constraint}		
	3\eps \leq \delta.
\end{constraint}
This constraint will be used in \Cref{subsec:DS-main} and \Cref{subsec:queries-main}.

The simplest constraints come from the fact that the thresholds for the classes should be in 
increasing order. 
So we have $1/3+\eps \leq 2/3-\eps$ implying the following constraint:
\begin{constraint}
	\label{eq:simple-constraint}		
	\eps \leq 1/6.
\end{constraint}

Solving these equations with the \emph{current best} value of the matrix multiplication exponent 
$\omega=2.371339$ gives us $\eps=0.0098109$ and $\delta=3\eps = 0.0294327$.
Using the \emph{best possible} value of the matrix multiplication exponent $\omega=2$ gives us 
$\eps=1/24$ and $\delta=1/8$.
\Cref{sec:apx-verify-constraints} shows that these values of $\eps$ and $\delta$ satisfy 
all the constraints.


\section{Main Algorithm}\label{sec:main-alg}

In this section, we prove \Cref{thm:alg-layer} under a few assumptions.
\AsmVertices*
\AsmClasses*

For vertices in an overlapping region i.e. they belong to two classes, put them in one of the classes 
arbitrarily. 
We will prove the following lemma:
\begin{lemma}\label{lem:layered-under-asm}
	 \Cref{thm:alg-layer} holds under \Cref{asm:vertices} and \Cref{asm:classes}.
\end{lemma}

\paragraph{High Level Idea.}
We would like to do something similar to the previous algorithm 
(\Cref{sec:alg-simple}) but here $A$ 
and $C$ are not fixed so there is a problem with aggregation of data structures over 
different 
\emph{chunks} when using the previous approach.
For instance, consider the data structure $A^{L*} \cdot B_{<i,DD}$ we maintained in \Cref{sec:alg-simple}. We maintained it by computing the matrix product of $A^{L*}$ with $B_{j,DD}$ for all $j < i$ and keeping a running sum. The running sum was updated in the following way: $A^{L*} \cdot B_{\leq i,DD} = A^{L*} \cdot B_{<i,DD} + A^{L*} \cdot B_{i,DD}$. Doing this is not possible if the matrix $A$ is undergoing changes.

Thus, we create different data structures for vertices in $H,M,L$ that can be 
updated on the fly i.e.\ 
that can be updated during each edge update which we show in 
\Cref{subsec:DS-main}. 
However, this is not enough, and we need one more idea to solve the problem.
Just like in \Cref{sec:alg-simple}, we need to use fast matrix multiplication to get a speedup.
We divide the edge updates into \emph{phases} of a certain size and during a 
phase we compute 
paths for the edges in the old phase whose edges are now fixed, in a way similar to 
what we did for 
\emph{chunks}.
We explain the idea of phases and its difference from chunks in 
\Cref{subsec:phases}.
Lastly, we show how we use these data structures and what we compute for 
different phases to 
answer queries in \Cref{subsec:queries-main}.

\subsection{Phases}\label{subsec:phases}
We define a \textbf{phase} as $m^{1-\delta}$ edge updates. 
A phase should be long enough so that in the time it takes to process all the edge 
updates in a 
phase, we are able to multiply two square matrices of dimension $m^{2/3+2\eps}$.
In other words, the number of edges in a phase times the worst-case update time 
should be more 
than the time it takes to multiply two square matrices of dimension 
$m^{2/3+2\eps}$ using fast 
matrix multiplication.
The amount of time it takes to do the multiplication using fast matrix multiplication is 
$m^{\omega 
\cdot (2/3 +2\eps)}$.
Thus, we need $m^{1-\delta} \cdot m^{2/3-\eps} \geq m^{\omega \cdot (2/3+2\eps)}$ which we get 
from \Cref{eq:phase-constraint}.
We divide the incoming edge updates into phases $P_1,P_2,\ldots$ each containing $m^{1-\delta}$ 
edge updates.
This means that the number of vertices that belong to a certain class reduces if we restrict ourselves 
to a particular phase.
For instance, the number of high degree vertices (degree at least $m^{2/3-\eps}$) in a phase is at 
most $m^{1-\delta-2/3+\eps}$ as opposed to $m^{1-2/3+\eps}$.

A \emph{phase} feels similar to a \emph{chunk} (introduced in 
\Cref{sec:setup-simple}), but they 
have slightly different purposes apart from 
having different sizes.
A chunk is small enough for us to do lazy evaluation on and when it is being inserted we create some 
specific data structures for the previous chunk.
During a phase we can compute all types of paths between all pairs of vertices for the previous 
phase (we are not limited to just specific data structures).

Note that because of the way a phase is set up, during a phase we can compute all matrix products 
(for relevant submatrices) for the previous phase (for $A,B,C,D$ in that phase).
This gives us the number of paths between all pairs of vertices that go through certain classes of 
vertices (the exact types of paths we need will be clear when we discuss how we answer queries).

We divide all the phases into two parts $\pnew$ which contains the current phase 
$P_{j+1}$ and the phase just before it $P_{j}$. All phases older than that are represented by 
$\po$. 
The edges in $A$ that are part of $\po$ are represented by $\ao$ and those that are part of 
$\pnew$ are represented by $\anew$ (similar for $B,C$).
So the new phase has at most $2 m^{1-\delta}$ edges but the old phase can have $O(m)$ 
edges.

A noteworthy observation is that we require $\delta>0$, which happens only when $\omega < 2.5 
-O(\eps)$ (\Cref{eq:phase-constraint}) which in particular happens only when $\omega<2.5$.
As we mention in the introduction, this is very surprising because any upper bound on $\omega$ better than $3$ is not sufficient.

\subsection{Data Structures}\label{subsec:DS-main}

\begin{table}[h!]
	\centering
	\footnotesize 
	\begin{tabular}{ |p{3.2cm}|p{3.2cm}|p{3.2cm}| p{3.2cm}| }
		\hline
		\multicolumn{4}{|c|}{\textbf{Data Structures}} \\
		\hline\hline	
		\textbf{High}& \textbf{Medium} &\textbf{Low} & \textbf{Common}  \\
		\hline
			 $A^{HD} \cdot B^{DD}$ (\Cref{eq:DS-HM-main}) & $A^{MD} \cdot 
			 B^{DD}$(\Cref{eq:DS-HM-main}) & $\anew^{*D} \cdot 
			 \bo^{DD}$ (\Cref{eq:DS-L-main}) & $A^{*S} \cdot 
		B^{S*}$ (\Cref{eq:DS-com-main})	\\
		\hline
		$B^{DD} \cdot C^{DH}$ (\Cref{eq:DS-HM-main}) & $B^{DD} \cdot C^{DM}$ (\Cref{eq:DS-HM-main}) 
		& $\bo^{DD} \cdot 
		\cnew^{D*}$ (\Cref{eq:DS-L-main}) & $B^{*S} \cdot C^{S*}$ (\Cref{eq:DS-com-main}) \\
		\hline
		\hline
		\multicolumn{4}{|c|}{\textbf{Additional Data Structures for High (\Cref{eq:DS-H-main}) }} \\		
		\hline
		\multicolumn{4}{|c|}{
		$\anew^{HS}\cdot \bnew^{SS}\cdot \cnew^{SH}, \;\;\;\;
		\anew^{HS}\cdot \bnew^{SS}\cdot \co^{SH}, \;\;\;\;
		\ao^{HS}\cdot \bnew^{SS}\cdot \cnew^{SH},\;\;\;\;
		\anew^{HS}\cdot \bo^{SS}\cdot \cnew^{SH}, $}\\
		\multicolumn{4}{|c|}{
		$\anew^{HS}\cdot \bo^{SS}\cdot \co^{SH}, \;\;\;\;
		\ao^{HS}\cdot \bo^{SS}\cdot \cnew^{SH}$} \\
		\hline
	\end{tabular}
\caption{Data Structures used in the Main Algorithm.}
\label{tab:DS-main}
\end{table}

We discuss the data structures we store to answer the queries in 
$O(m^{2/3-\eps})$ time. 
We will show that the worst-case time it takes to update the data structures is $O(m^{2/3-\eps})$.
The list of data structures can be found in \Cref{tab:DS-main}.

\begin{lemma}\label{lem:DS-update-time}
	The worst-case time it takes to update all the data structures we store is $O(m^{2/3-\eps})$.
\end{lemma}

Consider one of these data structures, for instance $\anew^{HS}\cdot \bnew^{SS}\cdot \co^{SH}$. It can be expressed in database terminology as follows:  
We start with the binary relation $A$ and filter the tuples to include only those belonging to the new phase.
We then select the tuples where the first attribute corresponds to a high-degree vertex in $L_1$ and the second attribute corresponds to a sparse vertex in $L_2$, denoted as $\anew^{HS}$. 
For the binary relation $B$, we similarly filter the tuples to include only those from the new phase and then select the tuples where both attributes correspond to sparse vertices in $L_2$ and $L_3$, denoted as $\bnew^{SS}$. 
For the binary relation $C$, we filter the tuples to include only those from the old phase and then select the tuples where the first attribute corresponds to a sparse vertex in $L_3$ and the second to a high-degree vertex in $L_4$, denoted as $\co^{SH}$.
We then perform a join between the filtered relations $\anew^{HS}$, $\bnew^{SS}$, and $\co^{SH}$ on their shared attributes (vertices in $L_2$ and $L_3$), forming three-hop paths.  
Finally, we aggregate the result over pairs of vertices in $L_1$ and $L_4$, by computing the count of three-hop paths between them.

\paragraph{Common Data Structures.}
We first discuss the data structures we store for all the vertices.
We maintain the data structures: 
\begin{equation}\label{eq:DS-com-main}
A^{*S} \cdot B^{S*} \text{ and } \; B^{*S} \cdot C^{S*}.	
\end{equation}

\begin{claim}\label{clm:S-DS}
	The worst-case update time for the data structures $A^{*S} \cdot B^{S*}$ and 
	$B^{*S} \cdot 
	C^{S*}$ is $O(m^{2/3-\eps})$.
\end{claim}
\begin{proof}
	Consider the data structure $A^{*S} \cdot B^{S*}$. The analysis for the other is 
	identical.
	If there is an update $(u,v)$ in $A^{*S}$ then we go over the $2m^{2/3-\eps}$ neighbors of $v$ 
	in $L_3$ and update the number of paths in the data structure. 
	If the update $(u,v)$ is in $B^{S*}$ then we iterate 
	over the $2m^{2/3-\eps}$ neighbors of $u$ in $L_1$ and update the number of paths.
\end{proof}

We also store the following data structures for all vertices, but we use them only for 
low degree 
vertices:
\begin{equation}\label{eq:DS-L-main}
\anew^{*D} \cdot
\bo^{DD} \text{ and } \; \bo^{DD} \cdot \cnew^{D*}.
\end{equation}

\begin{claim}\label{clm:L-DS}
	The worst-case update time for data structures $\anew^{*D} \cdot 
	\bo^{DD},$ $\bo^{DD} 
	\cdot \cnew^{D*}$ is $O(m^{1/3+\eps})$.
\end{claim}
\begin{proof}
	Observe that we only deal with $\bo$ and not $\bnew$ so there are no updates in $B$.
	If there is an update $(u,v)$ in $A$ (or $C$) then we iterate over the dense neighbors of $v$ in $L_3$ 
	in $O(m^{1/3+\eps})$ time and update the data structure.
\end{proof}

\paragraph{Data Structures for vertices in H,M.}

We now mention the data structures we store for the high and medium degree vertices 
in $L_1$ and $L_4$.
We maintain the following data structures: 
\begin{equation}\label{eq:DS-HM-main}
A^{HD} \cdot B^{DD}, \;\; B^{DD} \cdot C^{DH} \text{ and } \; A^{MD} \cdot B^{DD}, \;\;
B^{DD} \cdot C^{DM}.
\end{equation}

Notice that we can always update these data structures for any edge update in $O(m^{2/3-\eps})$ 
time. 
The reason for this is that the maximum size of these restricted layers is at most $m^{2/3-\eps}$.
\begin{claim}\label{clm:HM-DS}
	The worst-case update time for the data structures in 
	\Cref{eq:DS-HM-main} is 
	$O(m^{2/3-\eps})$.
\end{claim}
\begin{proof}
	Consider the data structure $A^{MD} \cdot B^{DD}$.
	If there is an update $(u,v)$ in $A^{MD}$ then we go over the $m^{1/3+\eps}$ neighbors of $v$ in 
	$L_3^D$ and update the number of paths in the data structure. If the update is in $B^{DD}$ then 
	we iterate over the $m^{2/3-\eps}$ vertices in $L_1^M$ and update the number of paths.
	The analysis is similar for $ B^{DD} \cdot C^{DM}$.
	The same arguments work for high degree vertices. This concludes the proof.
\end{proof}

We maintain six additional data structures for high degree vertices: 
\begin{equation}\label{eq:DS-H-main}
\begin{split}
\anew^{HS}\cdot \bnew^{SS}\cdot \cnew^{SH} \; , \quad
\anew^{HS}\cdot \bnew^{SS}\cdot \co^{SH} \; , \quad
\ao^{HS}\cdot \bnew^{SS}\cdot \cnew^{SH} \; , \\
\anew^{HS}\cdot \bo^{SS}\cdot \cnew^{SH} \; , \quad
\anew^{HS}\cdot \bo^{SS}\cdot \co^{SH} \; , \quad
\ao^{HS}\cdot \bo^{SS}\cdot \cnew^{SH}. 	
\end{split}
\end{equation}

This list includes all possible 
eight combinations of old and new phases \emph{other than} $\ao^{HS} \cdot \bnew^{SS} \cdot 
\co^{SH}$ and 
$\ao^{HS} \cdot
\bo^{SS} \cdot \co^{SH}$.

\begin{claim}\label{clm:H-DS}
	The worst-case update time for the additional data structures for high degree vertices 
	in \Cref{eq:DS-H-main} is 
	$O(m^{2/3-\eps})$.
\end{claim}
\begin{proof}
	We can store the data structures of the form $A^{HS} \cdot 
	B^{S*}$ and $B^{*S} 
	\cdot C^{SH}$ (\Cref{eq:DS-com-main}) for all combinations of old and new phases using the same proof 
	as \Cref{clm:S-DS}.
	
	Say we get an edge update in $A$ (resp. $C$) then we can go over all vertices $m^{1/3+\eps}$ in 
	$L_4^H$ (resp. $L_1^H$) and update the number of paths. 
	This is possible since we have the number of $B^{*S} \cdot C^{SH}$ (resp. 
	$A^{HS} \cdot 
	B^{S*}$) paths stored.
	
	If we get an edge update in $B$ then we iterate over all pairs of vertices in $L_1^H$ and $L_4^H$ 
	and update the data structure.
	This seems like it will take too much time on the surface, but we know that the data structure has 
	either $A$ or $C$ in the new phase and hence the number of high degree vertices in a new phase 
	is at most $m^{1-\delta -2/3 +\eps}$ as opposed to $m^{1-2/3+\eps}$.
	Thus, iterating over all pairs of high vertices takes time $m^{1/3+\eps} \cdot m^{1-\delta - 2/3+\eps} \cdot 
	O(1) \leq O(m^{2/3-\eps})$ (\Cref{eq:iteration-pairs-constraint}).
\end{proof}

\begin{proof}[Proof of \Cref{lem:DS-update-time}]
	\Cref{clm:S-DS,clm:L-DS,clm:HM-DS,clm:H-DS} together prove \Cref{lem:DS-update-time}.
\end{proof}

\subsection{Queries}\label{subsec:queries-main}

\begin{figure}[t]
	\centering
    \resizebox{150pt}{90pt}{ \begin{tikzpicture}
	\node[vertexBlack] (a1) at (0,0)  {};
	\node[vertexBlack] at (0,1) (a2) {};
	\node[vertexBlack] at (0,2) (a3) {};
	\node[vertexRed] at (0,3) (a4) {};
 
	\node[vertexBlack] at (2,0) (b1) {};
	\node[vertexBlack] at (2,1) (b2) {};
	\node[vertexBlack] at (2,2) (b3) {};
    \node[vertexBlack] at (2,3) (b4) {};

	\node[vertexBlack] at (4,0) (c1) {};
	\node[vertexBlack] at (4,1) (c2) {};
	\node[vertexBlack] at (4,2) (c3) {};
	\node[vertexBlack] at (4,3) (c4) {};

	\node[vertexBlack] at (6,0) (d1) {};
	\node[vertexBlack] at (6,1) (d2) {};
	\node[vertexBlack] at (6,2) (d3) {};
	\node[vertexRed] at (6,3) (d4) {};

	\draw[edge,bend left=40,dashed] (a4) to (d4);

    \draw[edge,blue] (c3) to (b1);
    \draw[edge,blue] (c3) to (b2);
    \draw[edge,blue] (a4) to (b2);
    \draw[edge,blue] (a4) to (b1);
    \draw[edge,red] (d4) to (c3);

    \draw[edge,red] (d4) to (c1);
    \draw[edge,blue] (b1) to (c1);

    \draw[edge,red] (d4) to (c4);

    \draw[edge,green] (b3) to (c4);
    \draw[edge,green] (b3) to (a4);

\draw (1,-0.5) node [anchor=north west][inner sep=0.75pt]   [align=left] 
{\normalsize{$A$}};
	
\draw (3,-0.5) node [anchor=north west][inner sep=0.75pt]   [align=left] 
{\normalsize{$B$}};

\draw (5,-0.5) node [anchor=north west][inner sep=0.75pt]   [align=left] 
{\normalsize{$C$}};

\draw (-1.3,3) node [anchor=north west][inner sep=0.75pt]   [align=left] 
{\normalsize{$u \in H$}};

\draw (6.3,3.2) node [anchor=north west][inner sep=0.75pt]   [align=left] 
{\normalsize{$v\in M$}};

\draw   (1.7,3.3) -- (2.3,3.3) -- (2.3,1.7) -- (1.7,1.7) -- cycle ;

\draw (2.4,3.3) node [anchor=north west][inner sep=0.75pt]   [align=left] 
{\normalsize{$D$}};

\draw   (3.7,3.3) -- (4.3,3.3) -- (4.3,1.7) -- (3.7,1.7) -- cycle ;

\draw (4.4,3.3) node [anchor=north west][inner sep=0.75pt]   [align=left] 
{\normalsize{$D$}};

\draw   (3.7,-0.7) -- (4.3,-0.7) -- (4.3,1.5) -- (3.7,1.5) -- cycle ;

\draw (4.4,1.7) node [anchor=north west][inner sep=0.75pt]   [align=left] 
{\normalsize{$S$}};

\draw   (1.7,-0.7) -- (2.3,-0.7) -- (2.3,1.5) -- (1.7,1.5) -- cycle ;

\draw (2.4,1.7) node [anchor=north west][inner sep=0.75pt]   [align=left] 
{\normalsize{$S$}};

\draw (-0.2,3.8) node [anchor=north west][inner sep=0.75pt]   [align=left] 
{\normalsize{$L_1$}};
\draw (1.8,3.8) node [anchor=north west][inner sep=0.75pt]   [align=left] 
{\normalsize{$L_2$}};
\draw (3.8,3.8) node [anchor=north west][inner sep=0.75pt]   [align=left] 
{\normalsize{$L_3$}};
\draw (5.8,3.8) node [anchor=north west][inner sep=0.75pt]   [align=left] 
{\normalsize{$L_4$}};

\end{tikzpicture} }	
	\caption{This figure shows a query $(u,v)$ where $u\in H$ and $v\in M$. We iterate over the neighbors of 
	$v$ (shown in \textcolor{red}{red}) and then use the wedges through sparse vertices to $u$ stored in the 
	data structure 
	$A^{*S}B^{S*}$ (shown in \textcolor{blue}{blue}). To get the $DD$ paths we use the wedges through 
	dense vertices to $u$ 
	stored in the data structure $A^{HD}B^{DD}$ (shown in \textcolor{green}{green}). To get the DS paths we 
	need to iterate 
	over the neighbors of $u$ in $L_2^D$ and use the wedges through sparse vertices to $v$ stored in the 
	data structure $B^{*S}C^{S*}$ (not shown here).}	
	\label{fig:HM-query-main}
\end{figure}
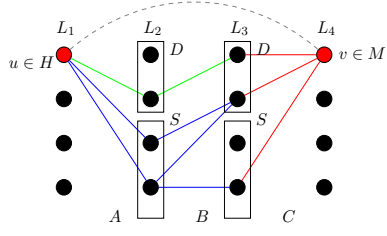

We now discuss how we can answer queries in worst-case time $O(m^{2/3-\eps})$.
We will prove the following:
\begin{lemma}\label{lem:query-time}
	The worst-case query time is $O(m^{2/3-\eps})$.
\end{lemma}

Let the current query be $e=(u,v)$ where $u$ is in $L_1$ and $v$ is in $L_4$.
We first discuss the queries when the endpoints belong to $H$ or $M$.
\begin{claim}\label{clm:HM-query-time}
	The worst-case query time when both the query endpoints belong to $H$ or $M$ is 
	$O(m^{2/3-\eps})$.
\end{claim}
\begin{proof}
There are two choices for vertices in $L_2$ and $L_3$, sparse or dense, so there are four kinds of 
paths between any pair of vertices in $L_1$ and $L_4$.
Consider paths where there is at least one dense vertex. 
We can iterate over all the $m^{1/3+\eps}$ dense vertices in that layer and check for an edge to 
one side and then the number of paths to the other side.
This can be done using the data structures of the form $A^{HD} \cdot B^{DD}$ 
(\Cref{eq:DS-HM-main}) 
for 
paths through dense vertices and $A^{*S} \cdot B^{S*}$ (\Cref{eq:DS-com-main}) 
for paths through 
sparse 
vertices.

So now we only worry about paths that go through $L_2^S$ and $L_3^S$.
If one of the vertices in the query is in $M$, let it be in $L_1$ without loss of generality, then we iterate over its 
$2m^{2/3-\eps}$ neighbors in $L_2^S$ 
and use $B^{*S} \cdot C^{S*}$ (\Cref{eq:DS-com-main}) to get the number of paths to the other endpoint of 
the query that go through $L_3^S$.

The remaining case is where the query has both high degree endpoints, and we need to calculate 
the number of $SS$ paths.
We know the number of $\ao \cdot \bo \cdot \co$ paths from the matrix product of the previous phase.
We also know the number of $\ao \cdot \bnew \cdot \co$ paths because we can run the 
algorithm where 
$A,C$ are fixed (\Cref{sec:alg-simple}).
For the remaining combinations of phases, we have data structures of the form 
$A^{HS} \cdot 
B^{SS} \cdot C^{SH}$ (\Cref{eq:DS-H-main}) where the number of paths are 
explicitly stored.
\Cref{fig:HM-query-main} gives an illustration of an $HM$ query.
\end{proof}

We now discuss the queries when at least one endpoint belong to $L$.
\begin{claim}\label{clm:L-query-time}
	The worst-case query time when at least one of the query endpoints belongs to $L$ is 
	$O(m^{2/3-\eps})$.
\end{claim}
\begin{proof}
We first consider $HL$ and $ML$ queries.
In this case, $u \in L_1$ is a high or medium degree vertex and $v \in L_4$ is a low degree vertex.
We can iterate over all $2m^{1/3 +\eps}$ neighbors $w_i$ of $v$. We know the number of $DD, SS, 
SD$ paths from $u$ to $w_i$'s using the data structures $A^{*S} \cdot B^{S*}$  
(\Cref{eq:DS-com-main})
and 
$A^{HD} \cdot B^{DD},A^{MD} \cdot B^{DD}$ (\Cref{eq:DS-HM-main}). 
For the $DS$ paths we iterate over the $2m^{1/3+\eps}$ neighbors of $u$ in $L_2^D$ and use 
$B^{*S} \cdot C^{S*}$ (\Cref{eq:DS-com-main}) to get the number of paths to $v$ 
via $L_3^S$.

We now consider $LL$ queries.
In this case, $u \in L_1^L$ and $v\in L_4^L$ are both low degree vertices.
It is easy to find the number of $SS,SD,DS$ paths by going over the neighbors of a query vertex and 
then using the data structures $A^{*S} \cdot B^{S*}, B^{*S} \cdot C^{S*}$ 
(\Cref{eq:DS-com-main}).
So we now only worry about $DD$ paths.
Consider the different cases depending on the phases $A,B$ and $C$ are in:

\textbf{Case 1:} $\ao \cdot \bo \cdot \co, \,\, \ao \cdot \bo \cdot \cnew, \,\, \anew \cdot 
\bo \cdot \co$ \\
Consider $\ao \cdot \bo \cdot \cnew$ paths (the others are similar).
We can iterate over the $2m^{1/3+\eps}$ neighbors of $v$ and use the number of 
$\ao \cdot \bo$ 
paths stored in the matrix product for the old phase.

\textbf{Case 2:} $\anew \cdot \bo \cdot \cnew$ \\
We have the following data structures stored: $\anew^{LD} \cdot \bo^{DD}, 
\bo^{DD} \cdot 
\cnew^{DL}$ (\Cref{eq:DS-L-main}).
Iterate over the $2m^{1/3+\eps}$ neighbors of $u$ and use the number of $DD$ paths to $v$.

\textbf{Case 3:} $\anew \cdot \bnew \cdot \cnew,\,$ 
$\anew \cdot \bnew \cdot \co,\,$ $\ao \cdot \bnew \cdot \cnew$ \\
Either $L_2$ or $L_3$ has just $O(m^{1-\delta})$ edges incident on it.
Thus, we go over all pairs of dense vertices in $L_2,L_3$ in time $O(m^{1/3+\eps}) \cdot O(m^{1/3-\delta 
+\eps}) 
\leq O(m^{2/3-\eps})$ (\Cref{eq:iteration-pairs-constraint}) and check if they form a path.

\textbf{Case 4:} $\ao \cdot \bnew \cdot \co$ \\
We can run the algorithm where $A,C$ are fixed (\Cref{sec:alg-simple}).

\end{proof}

\Cref{clm:HM-query-time} and \Cref{clm:L-query-time} together prove \Cref{lem:query-time}.
Also, \Cref{lem:DS-update-time} and \Cref{lem:query-time} together prove 
\Cref{lem:layered-under-asm}.
Therefore, we have proved \Cref{thm:alg-layer} under \Cref{asm:vertices} and \Cref{asm:classes}.
The next step is getting rid of these assumptions which we will do in the following sections.
The last thing we do in this section is give the pseudo code for the algorithms.

\subsection{Pseudocode}

In this subsection, we give some high level pseudocode for the main algorithm.
\Cref{Alg:main-count} is the main algorithm (Equivalent Queries in \Cref{sec:prelim}) which outputs the total number of $4$-cycles in the current graph after every update. 
It uses \Cref{Alg:main-update} and \Cref{Alg:main-query} as subroutines which are algorithms that update the data structures (\Cref{subsec:DS-main}) and answer the query (\Cref{subsec:queries-main}) respectively.
Note that these algorithms are only high level and the full details can be found in the referenced sections.

\begin{minipage}{0.45\textwidth}  
	
	\tiny{
		\begin{Algorithm}\label{Alg:main-count}
			\textbf{Counting layered $4$-Cycles}
			
			\vspace{4pt} 
			
			\textbf{Input:} A dynamic $4$-layered graph $G$ and an edge insertion or deletion $(u,v)$. \\
			\textbf{Output:} Maintain the count of total layered $4$-cycles in $G$.
			
			\vspace{4pt} 
			
			\textbf{Step 1: Process Edge Update $(u, v)$}
			\begin{itemize}[leftmargin=8pt]
				\item \textbf{Call Update($u,v$)} on the three relevant copies.
			\end{itemize}
			
			\vspace{4pt} 
			
			\textbf{Step 2: Compute Number of $4$-Cycles}
			\begin{itemize}[leftmargin=8pt]
				\item \textbf{Call Query($u,v$)} on the fourth copy:
				\[
				\text{New} \gets \textbf{Query}(u, v)
				\]
				\item Use the previous query answer to update the total count:
				\begin{align}
					\text{Total} &\gets \text{Total} + \text{New}  \tag{Insertion} \\
					\text{Total} &\gets \text{Total} - \text{New} \tag{Deletion}
				\end{align}
			\end{itemize}
			
			\textbf{Output} Total.

		\end{Algorithm}
	}
	
	\tiny{
		\begin{Algorithm}\label{Alg:main-update}
			\textbf{Update($u,v$)}
			
			\vspace{4pt} 
			
			\textbf{Input:} A dynamic graph $G$, an edge update $(u, v) \in B$, where $u \in L_2$ and $v \in L_3$. \\
			\textbf{Output:} Updates only the affected data structures containing $B$, reflecting the addition of $(u, v)$.
			
			\vspace{4pt} 
			
			\textbf{Step 1: Update Affected Data Structures:}  
			
			\textbf{If $u, v \in D$ :}  
			\begin{itemize}[leftmargin=8pt]
				\item Update data structures in \Cref{eq:DS-HM-main} by iterating over all neighbors of $u$ in $L_1^H$ or $L_1^M$ (similar for neighbors of $v$ in $L_4^H, L_4^M$).
			\end{itemize}
			
			\textbf{Else:}  
			\begin{itemize}[leftmargin=8pt]
				\item Update data structures in \Cref{eq:DS-com-main} by iterating over all neighbors of $u$ (similarly iterate over all neighbors of $v$).
				
				\item \textbf{If $u, v \in S$ :} Update data structures in \Cref{eq:DS-H-main} by iterating over all pairs of vertices in $L_1^H$ and $L_4^H$.
				
			\end{itemize}
			
			\vspace{4pt} 
			
			\textbf{Step 2: Matrix Multiplication for the Old Phase:}  
			
			\begin{itemize}[leftmargin=8pt]
				\item If $(u,v)$ is the first edge of a phase: \\
				We start the computation of matrix products for all relevant submatrices in the old phase.

				\item Else: \\
				Continue the matrix multiplication computation for $O(m^{2/3-\eps})$ steps.
				
			\end{itemize}

		\end{Algorithm}
	}
	
\end{minipage}%
\hfill  
\begin{minipage}{0.50\textwidth}  
	
	\tiny{
		\begin{Algorithm}\label{Alg:main-query}
			\textbf{Query($u,v$)}
			
			\vspace{4pt} 
			
			\textbf{Input:} A dynamic graph $G$, a query $(u, v)$ where $u \in L_1$ and $v \in L_4$. \\
			\textbf{Output:} Number of $3$-paths between $u$ and $v$.
			
			\vspace{4pt} 
			
			\textbf{If $u,v \in H \cup M$:}

			\begin{itemize}[leftmargin=8pt]
				\item Compute $\alpha$ the number of $DD, SD, DS$ paths by iterating over $L_2^D,L_3^D$ and using data structures of the form $A^{HD} \cdot B^{DD}$ for dense paths and $A^{*S} \cdot B^{S*}$ for sparse paths.
				
				\vspace{2pt} 
				
				\item \textbf{If} $u$ and $v \in H$, let $\beta$ be the sum of these paths:
				\begin{itemize} [leftmargin=4pt]
					\item Retrieve precomputed paths in $\ao \cdot \bo^{SS} \cdot \co$ and data structures in \Cref{eq:DS-H-main}.
					\item Query $u,v$ for the copy of the warm up algorithm run on $\ao \cdot \bnew^{SS} \cdot \co$.
				\end{itemize}
				
				\vspace{2pt} 
				
				\item \textbf{Else} (wlog $u\in M$) and let $\beta$ be the number of following paths.
				Iterate over neighbor of $u$ in $L_2^S$ and use $B^{*S} \cdot C^{S*}$.
				
			\end{itemize}
			
			\vspace{4pt} 
			
			\textbf{Else If $u,v \in L$:}
			
			\begin{itemize}[leftmargin=8pt]
				\item Compute $\alpha$ the number of $SS, SD, DS$ paths using $A^{*S} \cdot B^{S*}, B^{*S} \cdot C^{S*}$.
				\vspace{2pt} 
				\item Compute $\beta$ the number of $DD$ paths based on phases:
				\vspace{2pt} 
				\begin{itemize}[leftmargin=4pt]
					\item \textbf{Case 1:} $\ao \cdot \bo \cdot \co, \ao \cdot \bo \cdot \cnew, \anew \cdot \bo \cdot \co$ \\
					Iterate over neighbors in $D$ and use stored paths.
					\vspace{2pt} 
					
					\item \textbf{Case 2:} $\anew \cdot \bo \cdot \cnew$ \\
					Iterate over neighbors in $D$ and use data structures $\anew^{LD} \cdot \bo^{DD}$ and $\bo^{DD} \cdot \cnew^{DL}$.
					\vspace{2pt}
					
					\item \textbf{Case 3:} $\anew \cdot \bnew \cdot \cnew, \anew \cdot \bnew \cdot \co, \ao \cdot \bnew \cdot \cnew$ \\
					Iterate over all pairs of dense vertices in $L_2, L_3$ and check paths with $u,v$.
					\vspace{2pt}
					
					\item \textbf{Case 4:} $\ao \cdot \bnew \cdot \co$ \\
					Query $u,v$ for the copy of the warm up algorithm run for $\ao \cdot \bnew \cdot \co$.
				\end{itemize}
			\end{itemize}
			
			\vspace{4pt} 
			
			\textbf{Else:} (wlog $u\in L$)
			\begin{itemize}[leftmargin=8pt]
				\item Iterate over the neighbors of $u$.
				\item Compute $\alpha$ the number of $SS, SD, DS$ paths using $A^{*S} \cdot B^{S*}$ and $B^{*S} \cdot C^{S*}$.
				\item Compute $\beta$ the number of $DD$ paths using $A^{HD} \cdot B^{DD}, A^{MD} \cdot B^{DD}$.
			\end{itemize}
			
			\vspace{4pt} 
			
			\textbf{Return $\alpha+\beta$}.
			
		\end{Algorithm}
	}
	
\end{minipage}%

Note that \Cref{Alg:main-update} is written only for when the update is in matrix $B$. We have similar algorithms for when the update is in $A,C$ which update different data structures (using similar ideas). 
In Step $2$ of \Cref{Alg:main-update} we want to compute the number of paths between all pairs of vertices (between all layers) for all choices of classes of vertices $(S,D)$ in $L_2$ and $L_3$.
This is done by going over all possible $O(1)$ combinations of submatrices and multiplying the compatable ones using fast matrix multiplication. The matrix multiplications we compute are $A^{*x} \cdot B^{xy} \cdot C^{y*}, \, A^{*x} \cdot B^{x*}$ and $B^{*y} \cdot C^{y*}$ for all choices of $x,y \in \set{S,D}$.
In \Cref{Alg:main-query}, $\alpha$ and $\beta$ are just some counters where $\alpha$ counts $3$ types of paths and $\beta$ counts one type of path (there is no connection between $\alpha$ and $\beta$ across different cases).


\section{Getting Rid of \texorpdfstring{\Cref{asm:vertices}}{Assumption 1}: Tiny Vertices}
\label{sec:tiny-vertices}

In this section, we introduce \textbf{tiny} vertices ($T$) and show how they help get 
rid of 
\Cref{asm:vertices}.
The vertices in $L_1$ (resp. $L_4$) are in $T$ when their degree in $A$ (resp. $C$) is at most 
$2 m^{1/3-2\eps}$.
The vertices in $L_2$ (resp. $L_3$) are in $T$ when their combined degree in $A,B$ (resp. $B,C$) 
is at most $2 m^{1/3-2\eps}$.

We describe how we calculate the number of paths that go through tiny vertices and also how we 
answer queries when at least one endpoint is a tiny vertex.
Once we do this we can completely ignore tiny vertices.
For queries with a tiny vertex we will know how to compute the answer and for other queries we will 
be able to figure out all the paths that go through tiny vertices, so the only thing left to do will be to 
calculate the number of paths that do not go through any tiny vertices.
This is precisely what we focused on in the previous section \Cref{sec:main-alg} 
(modulo 
\Cref{asm:classes}). 
We will prove the following lemma:
\begin{lemma}\label{lem:thm-under-asm2}
	\Cref{thm:alg-layer} holds under \Cref{asm:classes}.
\end{lemma}

\paragraph{High Level Idea.}
The goal of this section is to say that tiny vertices are easy to handle.
We will show that we can easily compute paths through tiny vertices and answer 
queries containing 
tiny vertices.
It is easy to maintain $2$-paths through tiny vertices, but we also show in 
\Cref{subsec:DS-tiny} that 
we can store some types of $3$-paths through them in worst-case update time 
$O(m^{2/3-\eps})$.
These are helpful for finding paths through tiny vertices when one of the query 
endpoints is a high 
degree vertex.
If both the endpoints are not high degree then we can easily iterate over the 
neighbors of each of 
the query vertices and then use the data structure for wedges through tiny vertices.
The complete argument for this part can be found in \Cref{subsec:paths-tiny}.

Answering queries when one of the endpoints is a tiny vertex is the easier part, and 
we show the 
details in \Cref{subsec:queries-tiny}.
If the other vertex is tiny or low degree then we just iterate over the neighbors of 
both the endpoints.
Otherwise, we go over the neighbors of the tiny vertex and use the stored wedges 
through tiny and 
sparse vertices.
The only remaining case is paths going through dense vertices and the crucial 
observation is that we 
can iterate over all the neighbors of a tiny vertex and all the dense vertices of a layer.
This turns out to be enough to solve the problem.

\subsection{Data Structure for Tiny Vertices}\label{subsec:DS-tiny}
\begin{table}[h!]
	\centering
	\begin{tabular}{ |p{4.5cm}|p{4.5cm}| p{4.5cm}| }
		\hline
		\multicolumn{3}{|c|}{\textbf{Data Structures}} \\
		\hline\hline	
		\textbf{High}& \textbf{Medium} & \textbf{Common}  \\
		\hline
		 $A^{HT} \cdot B^{TT} \cdot C^{TH}$ (\Cref{eq:DS-HM-tiny}) &$A^{MT} \cdot B^{TT} \cdot 
		 C^{TH}$ (\Cref{eq:DS-HM-tiny}) & $A^{*T} \cdot B^{T*}$ (\Cref{eq:DS-common-tiny}) \\
		\hline		 
		$A^{HT} \cdot B^{TS} \cdot C^{SH}$ (\Cref{eq:DS-H-tiny}) &$A^{HT} \cdot B^{TT} \cdot C^{TM}$ 
		(\Cref{eq:DS-HM-tiny})& 
		$B^{*T} \cdot C^{T*}$ (\Cref{eq:DS-common-tiny})\\
		\hline		 
		$A^{HS} \cdot B^{ST} \cdot C^{TH}$(\Cref{eq:DS-H-tiny}) && \\
		\hline
	\end{tabular}
\caption{Data Structures for Tiny Vertices.}
\label{tab:DS-tiny}	
\end{table}

We show the data structures we need to store to calculate the number of paths 
through tiny 
vertices and to answer queries with tiny vertices in them.
The list of data structures can be found in \Cref{tab:DS-tiny}.
\begin{lemma}\label{lem:DS-T-update-time}
	The worst-case time it takes to update all the data structures we store for tiny vertices is 
	$O(m^{2/3-\eps})$.
\end{lemma}

We first maintain the following data structures for tiny vertices: 
\begin{equation}\label{eq:DS-common-tiny}
A^{*T} \cdot B^{T*} \text{ and } B^{*T} \cdot C^{T*}.	
\end{equation}

\begin{claim}\label{clm:T-DS}
	The worst-case update time for the data structures $A^{*T} \cdot B^{T*}$ and 
	$B^{*T} \cdot 
	C^{T*}$ is 
	$O(m^{1/3-2\eps})$.
\end{claim}
\begin{proof}
	The proof is almost identical to the proof of \Cref{clm:S-DS}.
	Consider the data structure $A^{*T} \cdot B^{T*}$. The analysis for the other 
	case is identical.
	If there is an update $(u,v)$ in $A^{*T}$ then we go over the $2m^{1/3-2\eps}$ neighbors of $v$ 
	in $L_3$ and update the number of paths in the data structure. 
	If the update $(u,v)$ is in $B^{T*}$ then we iterate 
	over the $2m^{1/3-2\eps}$ neighbors of $u$ in $L_1$ and update the number of paths.
\end{proof}

We maintain some data structures just for high and medium degree vertices: 
\begin{equation}\label{eq:DS-HM-tiny}
A^{HT} \cdot B^{TT} \cdot C^{TH},
A^{MT} \cdot B^{TT} \cdot C^{TH}, A^{HT} \cdot B^{TT} \cdot C^{TM}.
\end{equation}

\begin{claim}\label{clm:HM-T-DS}
	The worst-case update time for the data structures $A^{HT} \cdot B^{TT} \cdot 
	C^{TH},$ 
	$A^{MT} \cdot B^{TT} \cdot C^{TH},$ $A^{HT} \cdot B^{TT} \cdot C^{TM}$ is 
	$O(m^{2/3-\eps})$.
\end{claim}
\begin{proof}
	Consider the data structure $A^{HT} \cdot B^{TT} \cdot C^{TM}$. The analysis 
	for the others is 
	identical.
	First we note that we already have the data structures $A^{*T} \cdot B^{T*}, 
	B^{*T} \cdot C^{T*}$ 
	stored.
	This gives us access to $A^{HT} \cdot B^{TT}$ and $B^{TT} \cdot C^{TM}$.
	
	If there is an update $(u,v)$ in $A^{HT}$ then we go over the $m^{2/3-\eps}$ vertices in 
	$L_4^M$ and update the number of paths using $B^{TT} \cdot C^{TM}$.
	Similarly, if there is an update $(u,v)$ in $C^{TM}$ then we go over the $m^{1/3+\eps}$ vertices 
	in $L_1^H$ and update the number of paths using $A^{HT} \cdot B^{TT}$.
	
	Finally, if there is an update $(u,v)$ in $B^{TT}$ then we go over the $2m^{1/3-2\eps}$ neighbors 
	of $u$ in $L_1^H$ and $v$ in $L_4^M$ and update the number of paths in the data structure. 
	The total number of pairs is $4 m^{2/3-4\eps}$, so we are done.
\end{proof}

We maintain additional data structures for high degree vertices: 
\begin{equation}\label{eq:DS-H-tiny}
A^{HT} \cdot B^{TS} \cdot C^{SH}, 
A^{HS} \cdot B^{ST} \cdot C^{TH}.
\end{equation}

\begin{claim}\label{clm:HT-DS}
	The worst-case update time for these data structures for high degree vertices (\Cref{eq:DS-H-tiny}) is 
	$O(m^{2/3-\eps})$.
\end{claim}
\begin{proof}
	Consider the data structure $A^{HT} \cdot B^{TS} \cdot C^{SH}$. The analysis 
	for $A^{HS} \cdot 
	B^{ST} \cdot C^{TH}$ is 
	identical.
	First we note that we already have data structures $A^{*T} \cdot B^{T*}$ and 
	$B^{*S} \cdot 
	C^{S*}$ 
stored.
	This gives us access to $A^{HT} \cdot B^{TS}$ and $B^{TS} \cdot C^{SH}$.
	
	If there is an update $(u,v)$ in $A^{HT}$ then we go over the $m^{1/3+\eps}$ vertices in 
	$L_4^H$ and update the number of paths using $B^{TS} \cdot C^{SH}$.
	Similarly, if there is an update $(u,v)$ in $C^{SH}$ then we go over the $m^{1/3+\eps}$ vertices 
	in $L_1^H$ and update the number of paths using $A^{HT} \cdot B^{TS}$.
	
	Finally, if there is an update $(u,v)$ in $B^{TS}$ then we go over the $2m^{1/3-2\eps}$ neighbors 
	of $u$ in $L_1^H$ and also the $m^{1/3+\eps}$ vertices in $L_4^H$ and update the number of 
	paths in the data structure. 
	The total number of pairs is $2 m^{2/3-\eps}$, so we are done.
\end{proof}

\Cref{clm:T-DS,clm:HM-T-DS,clm:HT-DS} together prove \Cref{lem:DS-T-update-time}.
We now show how we answer queries that contain tiny vertices.

\subsection{Queries for Tiny Vertices}\label{subsec:queries-tiny}
Here we show how to answer queries when at least one of the vertices in the query is a tiny vertex.
\begin{lemma}\label{lem:T-query-time}
	The worst-case query time when at least one of the vertices in the query is a tiny vertex is 
	$O(m^{2/3-\eps})$.
\end{lemma}
\begin{proof}
We divide the types of queries into two cases.

\textbf{Case 1: TT, TL.} 
If one vertex is in $T$ and the other is in $T$ or $L$ then we can iterate over both of their neighbors 
in time $2 m^{1/3-2\eps} \cdot 2 m^{1/3+\eps} \cdot O(1) = O(m^{2/3-\eps})$ and then calculate the number 
of paths.

\textbf{Case 2: TM, TH.}
If one vertex is in $T$ and the other is in $H,M$ then we can first calculate the number of paths that 
go 
through $L_3^D$.
This is done by iterating over the $2m^{1/3-2\eps}$ neighbors of $u$ and all $m^{1/3+\eps}$ 
vertices in $L_3^D$ and checking if they form a path.
The time taken is $O(m^{2/3-\eps})$.

We can iterate over the $2m^{1/3-2\eps}$ neighbors of $u$ and then use the data structure 
$B^{*S}\cdot C^{S*}$ (\Cref{eq:DS-com-main}) to get the number of paths going through 
$L_3^S$ and use 
$B^{*T} \cdot C^{T*}$ (\Cref{eq:DS-common-tiny}) to 
get the number of paths going through $L_3^T$.
Thus, we get all types of paths
\end{proof}

\subsection{Paths Through Tiny Vertices}\label{subsec:paths-tiny}
Now we need to show that we can easily calculate the number of paths through tiny vertices in 
$L_2$ and $L_3$ no matter what kind of query we get.
Once we do this we can ignore tiny vertices for the rest of the analysis.

\begin{lemma}\label{lem:paths-through-T}
	The worst-case time it takes to calculate the number of paths through tiny vertices in $L_2$ or 
	$L_3$ during an arbitrary query is $O(m^{2/3-\eps})$.
\end{lemma}
\begin{proof}
We first describe queries where there are no high degree vertices.
Let the query be $u,v$. 
We have the data structures $A^{*T} \cdot B^{T*}$ and $B^{*T} \cdot C^{T*}$ 
(\Cref{eq:DS-common-tiny}) 
stored.
We iterate over neighbors of $u$ in $O(m^{2/3-\eps})$ time and use $B^{*T} \cdot 
C^{T*}$ to get 
the 
paths that go to $v$ via $L_3^T$.
We can do the same side from the side of $v$ and get the paths through $L_2^T$.
Note that we will include the $TT$ paths just once.

We can now see why the above procedure does not work for high degree vertices.
High degree vertices have too many neighbors to iterate over.
So we handle these cases separately.

\textbf{Case 1: HH.}
We can easily get the number of $TT$ paths using $A^{HT} \cdot B^{TT} \cdot 
C^{TH}$ (\Cref{eq:DS-HM-tiny}) and $TS$ 
paths using 
$A^{HT} \cdot B^{TS} \cdot C^{SH}$ (\Cref{eq:DS-H-tiny}). 
Finally, for the number of $TD$ paths, we iterate over the $m^{1/3+\eps}$ neighbors of $v$ in 
$L_3^D$ and then use the data structure $A^{*T} \cdot B^{T*}$ (\Cref{eq:DS-common-tiny}) to get the 
number 
of paths to $u$ 
via 
$L_2^T$.
The other paths can be found symmetrically.

\textbf{Case 2: HM.}
We can easily get the number of $TT$ paths using $A^{HT} \cdot B^{TT} \cdot 
C^{TM}$ (\Cref{eq:DS-HM-tiny}).
For the number of $TD$ paths (and $DT$ paths) we iterate over the $m^{1/3+\eps}$ neighbors of 
$v$ in $L_3^D$ and then use the data structure $A^{*T} \cdot B^{T*}$ (\Cref{eq:DS-common-tiny}) to get 
the number of paths to $u$ 
via $L_2^T$.
To get the number of $ST$ and $TS$ paths we iterate over the $2m^{2/3-\eps}$ neighbors of $v$ 
and then use $A^{*S} \cdot B^{S*}$ (\Cref{eq:DS-com-main}) to get the number of $ST$ paths and use 
$A^{*T} \cdot B^{T*}$ (\Cref{eq:DS-common-tiny}) to get the number of $TS$ paths.

\textbf{Case 3: HL.}
To find the number of $*T$ paths we can iterate over the neighbors of neighbors of $v$ and 
check if 
they have an edge to $u$. This takes $2m^{1/3+\eps} \cdot 2m^{1/3-2\eps} \cdot O(1) =O(m^{2/3-\eps})$ 
time.
For the number of $T*$ paths we can iterate over the neighbors of $v$ and use 
$A^{*T} \cdot 
B^{T*}$ (\Cref{eq:DS-common-tiny}).
\end{proof}

This means that we can figure out the number of paths through tiny vertices, and we can also 
answer 
queries when they involve tiny vertices, so we can completely ignore tiny vertices in the rest of the 
analysis.
Thus, we can assume that each layer has at most $m^{2/3+2\eps}$ vertices 
because all 
the remaining vertices have degree at least $m^{1/3-2\eps}$.

\begin{proof}[Proof of \Cref{lem:thm-under-asm2}]
	We are given a layered graph with $n$ vertices in each layer, and we are still under 
	\Cref{asm:classes}.
	We know from \Cref{lem:DS-T-update-time} that we can maintain the data structures for tiny 
	vertices in $O(m^{2/3-\eps})$ time.
	We know from \Cref{lem:T-query-time} that answering queries which include tiny vertices can be 
	answered in time $O(m^{2/3-\eps})$.
	Finally, \Cref{lem:paths-through-T} tells us that we can calculate the number of paths through tiny 
	vertices in time $O(m^{2/3-\eps})$.
	
	Thus, the only thing left to do is to answer queries with no tiny vertices and for them only calculate 
	paths that go through no tiny vertices.
	This means that we can assume that each layer has at most $m^{2/3+2\eps}$ vertices because all 
	the remaining vertices have degree at least $m^{1/3-2\eps}$.
	
	Therefore, we can assume \Cref{asm:vertices} without loss of generality.
	Also, we have already proved \Cref{lem:layered-under-asm}.
	This should be enough to prove \Cref{lem:thm-under-asm2}.	
	
	We need a final piece here to clarify some details.
	The common data structures we have in \Cref{eq:DS-com-main} and \Cref{eq:DS-L-main} currently hold 
	only when $*$ is either $H,M$ or $L$.
	We will make a case that they hold even when $*$ is $T$.
	This is done by observing that the proof of their update times never used the size of the layers.
	
	Consider the data structure $A^{*S} \cdot B^{S*}$ in \Cref{eq:DS-com-main} (the proof for the other one is 
	symmetric).
	When an edge update in $A$ or $B$ is incident on a vertex $v \in L_2^S$ we go over the $m^{2/3-\eps}$ 
	neighbors of $v$ and update the data structure.
	Thus, it does not matter what the size of the layers is because the number of neighbors of $v$ is small.
	
	Consider the data structure $\anew^{*D} \cdot \bo^{DD}$ in \Cref{eq:DS-L-main} (the proof for the other 
	one 
	is symmetric).
	When an edge update in $\anew$ (there are no updates in $\bo$) is incident on a vertex $v \in L_2^D$ 
	we go over the $m^{1/3+\eps}$ vertices in $L_3^D$ and update the data structure.
	Thus, it does not matter what the size of the layers is because the number of vertices in $L_3^D$ is small.
	Therefore, the common data structures work even when $*$ is $T$.
\end{proof}


\section{Getting Rid of \texorpdfstring{\Cref{asm:classes}}{Assumption 2}: Vertices Moving Between Classes}
\label{sec:vertices-moving-classes}

In this section, we get rid of \Cref{asm:classes}, namely, handle the case when the vertices may 
move between classes due to the change in their degrees.

\subsection{Setup}
Note that until now we have assumed that the classes of vertices are fixed. 
But if the degree of vertices change then their classes will change too. 
So in this section, we describe what happens when the degrees of the vertices changes by a lot, 
and 
they have to change classes.


When the degree of a vertex is changing it goes from being in one class to the overlapping region of 
two classes and then finally into one class again.
When the vertex moves into to the overlapping region of two classes we treat it like it still belongs to 
the old class.
During this period we will create data structures of both classes for this vertex.
After undergoing edge insertions and deletions if the vertex ends up in its old class then we do not 
have to do anything.
However, if it ends up in a new class then we have to update the data structures for both the old 
and 
new classes.
Note that we have to address \emph{two types} of vertex transitions, one where the vertex is the endpoint 
of a path called \textbf{\To} transitions and the other where the vertex is an intermediate vertex in a path 
called \textbf{\Tt} transitions.

\paragraph{\To Transition.}
In this type of transition, the transitioning vertex is the endpoint of a path for some data structure.
For instance consider the data structure $A^{MT} \cdot B^{TT} \cdot C^{TH}$ (\Cref{eq:DS-HM-tiny}) for 
medium vertices that we do not have for low vertices.
Consider a vertex $u \in L_1$ with degree $m^{1/3+\eps}$ which after an edge insertion, has just entered the 
transition range between $L$ and $M$.
Computing this data structure for $u$ means that for all vertices $v_i \in L_4^H$ we need to 
compute the number of paths between $u$ and $v_i$ that go through $L_2^T$ and $L_3^T$.
This is the first type of transition and what we have to do for it is add a new row to the data 
structures for medium vertices like $A^{MT} \cdot B^{TT} \cdot C^{TH}$ corresponding to the 
answers we compute for $u$.
We also need to remove all the rows corresponding to $u$ in the data structures for low vertices.

\paragraph{\Tt Transition.}
In this type of transition, the transitioning vertex is an intermediate vertex in a path for some data structure.
Consider the same data structure $A^{MT} \cdot B^{TT} \cdot C^{TH}$ (\Cref{eq:DS-HM-tiny}) and consider 
a vertex $y \in L_2$ with degree $m^{1/3-2\eps}$ which after an edge insertion, has just 
entered the transition range between $T$ and $S$.
Computing this data structure for $y$ means that for all pairs of vertices $u_i \in L_1^M$ and $v_j \in 
L_4^H$ we need to compute the number of paths between $u_i$ and $v_j$ that go through $y$ 
in $L_2$ and through some vertex in $L_3^T$.
This is the second type of transition and what we have to do for it is update the values in data 
structures for medium vertices like $A^{MT} \cdot B^{TT} \cdot C^{TH}$ corresponding to the 
answers we compute for vertex $y$.
So this can be thought of computing a difference matrix of dimension $\card{L_1^M} \times 
\card{L_4^H}$ which when added to the matrix $A^{MT} \cdot B^{TT} \cdot C^{TH}$ includes the 
contribution of vertex $y \in L_2$.
We also have to update the values in data structures for low vertices.
The major difference between the two types of transitions is that for the first type of transition (\To), we 
need to add or delete a new row of values corresponding to the transitioning vertex $u$, and we 
do not have to change any existing values.
In the second type of transition (\Tt), we do not have to add or delete any rows but update the existing 
values in the data structures for all pairs that are affected by the transitioning vertex $y$.

Note that the computation to create data structures for a vertex $v$ will only be done when we 
get 
an edge incident to this vertex $v$ so by the time its degree doubles we have the new data 
structure 
ready for $v$.
This ensures that during an edge update we are doing this computation for at most two vertices 
and thus does not affect the worst-case update time (if we are not careful about this then the 
guarantee we 
get will be in terms of amortized update time).
Finally, when the vertex $v$ enters the transition range, we partition the edges into the ones that 
came 
before and after that moment and deal with them separately.
Let edge $e$ be the update after which vertex $v$ entered the transition range.
We represent by $\ebef^v$, the set of edge updates that arrive before edge update $e$ (including $e$) and 
by $\eaft^v$ set of edge updates that arrive after edge update $e$.
We will drop the superscript of $v$ when the vertex is clear from context.
We also denote by $\abef$ and $\aaft$ the set of edges in $A$ before and after the edge update $e$, so we 
have 
$A=\abef+\aaft$.
Note that $\aaft$ can have a negative value for some edge because that edge was inserted in $\abef$ and 
deleted in $\aaft$ but this is not going to be a problem.
We define similar notation for $B$ and $C$.
The main task is to compute the data structure for the edges that come before the transition
moment ($\ebef$) in $O(x \cdot m^{2/3-\eps})$ time, where $x$ is the degree of the vertex when it enters 
the overlapping region.

Consider any data structure we need to compute for the new class where edges of $\eaft$ are being used (in 
$A,B$ or $C$) for instance consider $\aaft^{MT} \cdot \bbef^{TT} \cdot \cbef^{TH}$.
Notice that their initial count is $0$ for all pairs because initially $\eaft$ is empty.
Since this is a data structure of the new class, it can be maintained in worst-case update time 
$O(m^{2/3-\eps})$. In this example we know how to maintain $A^{MT} \cdot B^{TT} \cdot C^{TH}$ in 
worst-case update time $O(m^{2/3-\eps})$ where $L_1$ has all medium degree vertices.

The task for $\aaft^{MT} \cdot \bbef^{TT} \cdot \cbef^{TH}$ is slightly different.
We need to maintain this data structure for a subset of vertices in that class i.e.\ the vertices in $L_1$ 
transitioning from 
$L$ to $M$ (note that this is at most the number of vertices in $M$).
Also, we need to maintain this data structure on a subset of the edges.
In this example we need to use edges of $A \cap \eaft$, $B \cap \ebef$, and $C \cap \ebef$. 
But all of these edges are known so the same algorithm that we use for class $M$ for the data structure 
$A^{MT} \cdot B^{TT} \cdot C^{TH}$ will work for the for class of vertices transitioning from $L$ to $M$ for 
the data structure $\aaft^{MT} \cdot \bbef^{TT} \cdot \cbef^{TH}$.

Thus, any data structure we need to maintain with at least one of $A,B,C$ being in $\eaft$ can be maintained 
in worst-case update time $O(m^{2/3-\eps})$.
We now only worry about computing the data structures for transitioning vertices where all the edges are in 
$\ebef$.
We will show that all such data structures can be updated in worst-case time $O(x \cdot m^{2/3-\eps})$ 
where $x$ is the 
degree of the vertex when it enters 
the overlapping region.
\begin{lemma}\label{lem:vertex-transition}
	For a vertex $v$ that is transitioning from one class to the other, the data structures for 
	$v$ can be computed in $O(x \cdot m^{2/3-\eps})$ time,
	where $x$ is the degree of $v$ when it enters the overlapping region.
\end{lemma}

Note that we also need to prove this for the algorithm where $A$ and $C$ are fixed 
(\Cref{sec:alg-simple}) because we use the simple algorithm as a subroutine in our main algorithm.
But this is trivial because there are no edge updates in $A$ and $C$ so vertices in $L_1$ and $L_4$ do not 
change classes.
Note that this algorithm only has classes for the vertices $L_1$ and $L_4$ and the classes for $L_2$ and 
$L_3$ are local to each chunk and are used to partition the edges of a chunk.
So we do not have any vertex transitions for vertices in $L_2$ and $L_3$.
Thus, we can trivially drop \Cref{asm:classes} for the algorithm where $A$ and $C$ are fixed.
Therefore, \Cref{lem:vertex-transition} is the last piece we need to prove 
\Cref{thm:alg-layer}.
\begin{proof}[Proof of \Cref{thm:alg-layer}]
	\Cref{lem:vertex-transition} allows us to assume wlog that vertices remain in 
	the same classes because if they change classes then we can easily update the data structures.
	Thus, we can now use \Cref{lem:thm-under-asm2}, which we have already proved, to conclude the proof.
\end{proof}

\subsection{Proof of \texorpdfstring{\Cref{lem:vertex-transition}}{Lemma}}

In this subsection, we only deal with matrices $\abef,\bbef$ and $\cbef$. 
We do not use $\aaft,\baft$ and $\caft$ at all.
Thus, to simplify notation we drop the subscript ``bef'' and represent $\abef,\bbef$ and $\cbef$ by $A,B$ 
and $C$ respectively.

We start with all different cases for \emph{\To} transitions where the transitioning vertex is an endpoint of a 
path in the 
data structure.

\noindent
\textbf{Case 1: $L$ and $M$} \\
The amount of time we have in this case is $m^{1/3+\eps} \cdot \, O(m^{2/3-\eps}) = O(m)$.
We will only focus on vertices in $L_1$ transitioning from $L$ to $M$ because there are no data structures 
just for $L$.
Note that these arguments will be similar for $L_4$.

\begin{claim}
	The amount of time it takes to compute the data structures $A^{MD} \cdot B^{DD}$ and $B^{DD}\cdot  
	C^{DM}$ (\Cref{eq:DS-HM-main}) for a vertex $v \in L_1$ transitioning from $L$ to $M$ is $O(m)$.
\end{claim}
\begin{proof}
	Consider the data structure $A^{MD} \cdot B^{DD}$ (the proof for the other one is similar).
	We can iterate over all neighbors of $v$ in $L_2^D$ and then go over all their neighbors in 
	$L_3^D$ to update the number of paths.
	This takes time $m^{1/3+\eps} \cdot m^{1/3+\eps} \cdot O(1)$ that is within our budget.
\end{proof}

\begin{claim}
	The amount of time it takes to compute the data structures $A^{MT} \cdot B^{TT} \cdot C^{TH} $ and 
	$A^{HT} \cdot B^{TT} \cdot C^{TM}$ (\Cref{eq:DS-HM-tiny}) for a vertex $v$ transitioning from $L$ to 
	$M$ is $O(m)$.
\end{claim}
\begin{proof}
	Consider the data structure $A^{MT} \cdot B^{TT} \cdot C^{TH}$ for a vertex $v\in L_1$.
	We can iterate over all neighbors of $v$ in $L_2^T$ and then go over all their neighbors in 
	$L_3^T$ and finally over their neighbors in $L_4^H$ to update the number of paths.
	This takes time $m^{1/3+\eps} \cdot m^{2(1/3-2\eps)} \cdot O(1)=O(m^{1-3\eps})$ that is within our 
	budget.
\end{proof}

\noindent
\textbf{Case 2: $M$ and $H$} \\
The amount of time we in this case have is $m^{2/3-\eps} \cdot O(m^{2/3-\eps})=O(m^{4/3-2\eps})$.
We focus on a vertex $v\in L_1$ transitioning from $M$ to $H$ or $H$ to $M$.

We first consider data structures $A^{HT} \cdot B^{TS} \cdot C^{SH},$ 
$A^{HS} \cdot B^{ST} \cdot C^{TH}$ (\Cref{eq:DS-H-tiny}). 
\begin{claim}
	The amount of time it takes to compute the data structures $A^{HT} \cdot B^{TS} \cdot C^{SH}$ and 
	$A^{HS} \cdot B^{ST} \cdot C^{TH}$ for a vertex 
	$v\in L_1$ transitioning from $M$ to $H$ is $O(m^{4/3-2\eps})$.
\end{claim}
\begin{proof}
	Consider the data structure $A^{HT} \cdot B^{TS} \cdot C^{SH}$ (the proof for the other is similar).
	We can go over all $m^{2/3-\eps}$ neighbors of $v$ in $L_2^T$ and all their $m^{1/3-2\eps}$ 
	neighbors in $L_3^S$ and finally all $m^{1/3+\eps}$ vertices in $L_4^H$ and update the 
number of 
	paths.
	The time taken is $m^{2/3-\eps} \cdot m^{1/3-2\eps} \cdot m^{1/3+\eps} = m^{4/3-2\eps}$.
	
	Consider the data structure $A^{HS} \cdot B^{ST} \cdot C^{TH}$.
	We can go over all $m^{2/3-\eps}$ neighbors $w_i$ of $v$ in $L_2^S$, but we cannot iterate 
over 
	their neighbors.
	So we also iterate over all $m^{1/3+\eps}$ vertices $u_i$ in $L_4^H$ and then use the data 
structure 
	$B^{*T} \cdot C^{T*}$ (\Cref{eq:DS-common-tiny}) to get the number of paths between $w_i$ and $u_j$ 
	through $L_3^T$. 
	This takes time $m^{2/3-\eps} \cdot m^{1/3+\eps} \leq m^{4/3-2\eps}$.
\end{proof}

We now consider data structures $A^{HT} \cdot B^{TT} \cdot C^{TH},
A^{MT} \cdot B^{TT} \cdot C^{TH}, A^{HT} \cdot B^{TT} \cdot C^{TM}$ (\Cref{eq:DS-HM-tiny}). 
\begin{claim}
	The amount of time it takes to compute the data structures $A^{HT} \cdot B^{TT} \cdot C^{TH},$ $A^{MT} 
	\cdot B^{TT} \cdot C^{TH},$ $A^{HT} \cdot B^{TT} \cdot C^{TM}$ for a vertex 
	$v\in L_1$ transitioning between $M$ and $H$ is $O(m^{4/3-2\eps})$.
\end{claim}
\begin{proof}
	Consider the data structure $A^{MT} \cdot B^{TT} \cdot C^{TH}$ for a vertex $v\in L_1$ (the analysis for 
	the others is the same).
	We can iterate over all neighbors of $v$ in $L_2^T$ and then go over all their neighbors in 
	$L_3^T$ and finally over their neighbors in $L_4^H$ to update the number of paths.
	This takes time at most $2m^{2/3-\eps} \cdot m^{2(1/3-2\eps)} \cdot O(1)= O(m^{4/3-5\eps})$ that is 
	within our budget.
\end{proof}

We finally consider data structures of the form $\anew^{HS}\cdot \bnew^{SS}\cdot \cnew^{SH}$ 
(\Cref{eq:DS-H-main}).
\begin{claim}\label{clm:H-main-trans-HM}
	The amount of time it takes to compute the data structures of the form $\ao^{HS} \bo^{SS} \cnew^{SH}$ in 
	\Cref{eq:DS-H-main} for a vertex 
	$v\in L_1$ transitioning from $M$ to $H$ is $O(m^{4/3-2\eps})$.
\end{claim}
\begin{proof}
	Consider the data structure $A^{HS}\cdot B^{SS}\cdot C^{SH}$.
	We do not store this data structure, we only store this when $A,B$ and $C$ are restricted to some phases.
	But for the purpose of analysis we show how to compute the data structure for $A^{HS}\cdot B^{SS}\cdot 
	C^{SH}$ and the analysis for all the data structures in \Cref{eq:DS-H-main} is similar.
	
	We can go over all $m^{2/3-\eps}$ neighbors $w_i$ of $v$ in $L_2^S$, but we cannot iterate 
	over 
	their neighbors.
	So we also iterate over all $m^{1/3+\eps}$ vertices $u_i$ in $L_4^H$ and then use the data 
	structure 
	$B^{*S} \cdot C^{S*}$ (\Cref{eq:DS-com-main}) to get the number of paths between $w_i$ and $u_j$ 
	through $L_3^S$. 
	This takes time $m^{2/3-\eps} \cdot m^{1/3+\eps} \cdot O(1) \leq O(m^{4/3-2\eps})$.
\end{proof}

Note that we do not have a case for transition between $T$ and $L$ because both of them do not have any 
data structures in $L_1$ (or $L_4$).

\noindent
\textbf{Case 3: $S$ and $D$} \\
The amount of time we in this case have is $m^{2/3-\eps} \cdot O(m^{2/3-\eps})=O(m^{4/3-2\eps})$.
We focus on a vertex $v\in L_3$ transitioning between $S$ and $D$.
This will be similar for vertices in $L_2$.

We only have to consider the data structures of the form $A^{MD}B^{DD}$ in \Cref{eq:DS-HM-main} and 
\Cref{eq:DS-L-main}.
\begin{claim}
	The amount of time it takes to compute the data structures of the form $A^{MD}B^{DD}$ in 
	\Cref{eq:DS-HM-main} and 
	\Cref{eq:DS-L-main} for a vertex $v$ transitioning from $S$ to $D$ is $O(m^{4/3-2\eps})$.
\end{claim}
\begin{proof}
	Consider the data structure $A^{MD}B^{DD}$ (the analysis for the others is similar) and let $v\in L_3$.
	We can iterate over the $m^{2/3-\eps}$ neighbors in $L_1$ and all $m^{1/3+\eps}$ vertices in $L_2^D$ 
	and create the data structure in time $m^{2/3-\eps} \cdot m^{1/3+\eps} \cdot O(1) \leq 
	O(m^{4/3-2\eps})$.
\end{proof}


We now discuss all different cases for \emph{\Tt} transitions where the transitioning vertex is an intermediate 
vertex of a path in the data structure.

%

\noindent
\textbf{Case 1: $S$ and $D$} \\
Consider vertices transitioning between $S$ and $D$.
The amount of time we have is $m^{2/3-\eps} \cdot O(m^{2/3-\eps}) = O(m^{4/3-2\eps})$.

We first consider the data structures of the form $A^{MD} \cdot B^{DD}$ in \Cref{eq:DS-HM-main} and 
\Cref{eq:DS-L-main}.
\begin{claim}
	The amount of time it takes to create the data structures in \Cref{eq:DS-HM-main} and 
	\Cref{eq:DS-L-main} for a vertex $v$ transitioning from $S$ to $D$ is $O(m) \leq O(m^{4/3-2\eps})$.
\end{claim}
\begin{proof}
	Consider the data structure $A^{MD}B^{DD}$ (the analysis for the others is similar) and let $v\in L_2$.
	We can iterate over the $m^{2/3-\eps}$ neighbors in $L_1$ and all $m^{1/3+\eps}$ vertices in $L_3^D$ 
	and create the data structure in time $m^{2/3-\eps} \cdot m^{1/3+\eps} \cdot O(1) =O(m) \leq 
	O(m^{4/3-2\eps})$.
\end{proof}

We now consider the data structures $A^{*S} \cdot B^{S*}$ and $B^{*S} \cdot C^{S*}$ 
(\Cref{eq:DS-com-main}).
\begin{claim}
	The amount of time it takes to create the data structure $A^{*S} \cdot B^{S*}$ and $B^{*S} \cdot C^{S*}$ 
	for a vertex $v$ transitioning from $D$ to $S$ is $O(m^{4/3-2\eps})$.
\end{claim}
\begin{proof}
	Consider the data structure $A^{*S} \cdot B^{S*}$ (the analysis for the other is similar) and let $v\in L_2$.	
	We go over all pairs of $2 m^{2/3-\eps}$ neighbors on both sides ($L_1$ and $L_3$) and create the data 
	structure.
\end{proof}

We next consider the data structures of the form $\ao^{HS} \bo^{SS} \cnew^{SH}$ (\Cref{eq:DS-H-main}).
\begin{claim}\label{clm:H-main-trans-SD}
	The amount of time it takes to create the data structures of the form $\ao^{HS} \bo^{SS} \cnew^{SH}$ in 
	\Cref{eq:DS-H-main} 
	for a vertex $v\in L_2$ transitioning from $D$ to $S$ is $O(m^{4/3-2\eps})$.
\end{claim}
\begin{proof}
	Consider the data structure $A^{HS}\cdot B^{SS}\cdot C^{SH}$.
	We do not store this data structure, we only store this when $A,B$ and $C$ are restricted to some phases.
	But for the purpose of analysis we show how to compute the data structure for $A^{HS}\cdot B^{SS}\cdot 
	C^{SH}$ and the analysis for all the data structures in \Cref{eq:DS-H-main} is similar.
	We did the same thing in the proof of \Cref{clm:H-main-trans-HM}.
	
	We go over all $2 m^{2/3-\eps}$ neighbors of $v$ in $L_3^S$ and also all pairs of high degree 
	vertices in $L_1$ and $L_4$ and create the data structure.
	On the surface this seems like it takes too much time but note that all the relevant data 
structures of this type have either $A$ or $C$ in the new phase, and so we can iterate over all pairs of high 
	degree vertices in 
	time $O(m^{2/3-\eps})$.
	This is similar to the proof of \Cref{clm:H-DS} and thus, gives us the desired time bound.
\end{proof}

Finally, consider the data structures $A^{HS} \cdot B^{ST} \cdot C^{TH}$ and $A^{HT} \cdot B^{TS} \cdot 
C^{SH}$ (\Cref{eq:DS-H-tiny}).
\begin{claim}
	The amount of time it takes to create the data structures $A^{HS} \cdot B^{ST} 
	\cdot C^{TH}$ and $A^{HT} \cdot B^{TS} \cdot C^{SH}$
	for a vertex $v$ transitioning from $D$ to $S$ is $O(m^{2/3+2\eps}) \leq O(m^{4/3-2\eps})$.
\end{claim}
\begin{proof}
	Consider the data structure $A^{HS} \cdot B^{ST} 
	\cdot C^{TH}$ (the analysis for the other is similar) and let $v\in L_2$.
	We go over all $m^{1/3+\eps}$ neighbors of $v$ in $L_1^H$ and also all $m^{1/3+\eps}$ high 
	degree vertices $w_i$ in $L_4$. 
	We then use the number of paths between $v$ and $w_i$ through $L_3^T$ stored in 
	$B^{*T} \cdot C^{T*}$ (\Cref{eq:DS-common-tiny}) and create the data structure.	
	This gives us the desired time bound of $O(m^{2/3+2\eps}) \leq O(m^{4/3-2\eps})$.
\end{proof}

\noindent
\textbf{Case 2: $T$ and $S$} \\
Note that both $T,L$ do not have any data structures in $L_1,L_4$, so we do not have to do anything for a 
transition between $T$ and $L$.
We first consider the transition from $T$ to $S$.
The amount of time we have is $m^{1/3-2\eps} \cdot O(m^{2/3-\eps}) = O(m^{1-3\eps})$.

We start with the data structures $A^{*S} \cdot B^{S*}$ and $B^{S*} \cdot C^{*S}$ (\Cref{eq:DS-com-main}).
\begin{claim}
	The amount of time it takes to create the data structures $A^{*S} \cdot B^{S*}$ and $B^{S*} \cdot C^{*S}$ 
	for a vertex $v$ transitioning from $T$ to $S$ is $O(m^{2/3-4\eps})\leq O(m^{1-3\eps})$.
\end{claim}
\begin{proof}
Consider the data structure $A^{*S} \cdot B^{S*}$ (the analysis for the other is similar) and let $v\in L_2$.	
We go over all pairs of $m^{1/3-2\eps}$ neighbors on both sides ($L_1$ and $L_3$) and create the data 
structure.
This takes time $m^{1/3-2\eps} \cdot m^{1/3-2\eps} \cdot O(1) = O(m^{2/3-4\eps}) \leq 
O(m^{1-3\eps})$.
\end{proof}

We now consider the data structures of the form $\ao^{HS} \bo^{SS} \cnew^{SH}$ (\Cref{eq:DS-H-main}).
\begin{claim}
	The amount of time it takes to create the data structures of the form $\ao^{HS} \bo^{SS} \cnew^{SH}$ 
	in \Cref{eq:DS-H-main}
	for a vertex $v \in L_2$ transitioning from $T$ to $S$ is $O(m^{1-3\eps})$.
\end{claim}
\begin{proof}
		Consider the data structure $A^{HS}\cdot B^{SS}\cdot C^{SH}$.
	We do not store this data structure, we only store this when $A,B$ and $C$ are restricted to some phases.
	But for the purpose of analysis we show how to compute the data structure for $A^{HS}\cdot B^{SS}\cdot 
	C^{SH}$ and the analysis for all the data structures in \Cref{eq:DS-H-main} is similar.
	We did the same thing in the proof of \Cref{clm:H-main-trans-HM} and \Cref{clm:H-main-trans-SD}.
	
	We go over all $m^{1/3-2\eps}$ neighbors of $v$ in $L_3^S$ and also all pairs of high degree 
	vertices in $L_1$ and $L_4$ and create the data structure.
	This seems like it takes too much time but note that all the relevant data 
structures 
	of this type have either $A$ or $C$ in the new phase, and so we can iterate over all pairs of high 
	degree vertices in time $O(m^{2/3-\eps})$.
	This gives us the desired time bound of $m^{1/3-2\eps} \cdot O(m^{2/3-\eps}) = O(m^{1-3\eps})$.
\end{proof}

Finally, consider the data structures $A^{HS} \cdot B^{ST} \cdot C^{TH}$ and $A^{HT} \cdot B^{TS} \cdot 
C^{SH}$ (\Cref{eq:DS-H-tiny}).
\begin{claim}
	The amount of time it takes to create the data structures $A^{HS} \cdot B^{ST} 
	\cdot C^{TH}$ and $A^{HT} \cdot B^{TS} \cdot 
	C^{SH}$ 
	for a vertex $v$ transitioning from $T$ to $S$ is $O(m^{2/3+2\eps}) \leq O(m^{1-3\eps})$.
\end{claim}
\begin{proof}
	Consider the data structure $A^{HS} \cdot B^{ST} \cdot C^{TH}$ (the analysis for the other is similar) and 
	let $v\in L_2$.
	We go over all $m^{1/3+\eps}$ neighbors of $v$ in $L_1^H$ and also all $m^{1/3+\eps}$ vertices $w_i$ in 
	$L_4^H$. 
	We then use the number of paths between $v$ and $w_i$ through $L_3^T$ stored in 
	$B^{*T} \cdot C^{T*}$ (\Cref{eq:DS-common-tiny}) and create the data structure.	
	This gives us the desired time bound of $O(m^{2/3+2\eps}) \leq O(m^{1-3\eps})$.
\end{proof}

Now consider vertices transitioning from $S$ to $T$.
The amount of time we have is $2 m^{1/3-2\eps} \cdot O(m^{2/3-\eps})=O(m^{1-3\eps})$.
We start with the data structures $A^{*T} \cdot B^{T*}$ and $B^{T*} \cdot C^{*T}$ 
(\Cref{eq:DS-common-tiny}).
\begin{claim}
	The amount of time it takes to create the data structures $A^{*T} \cdot B^{T*}$ and $B^{T*} \cdot C^{*T}$ 
	for a vertex 
	$v$ transitioning from $S$ to $T$ is $O(m^{2/3-4\eps}) \leq O(m^{1-3\eps})$.
\end{claim}
\begin{proof}
	Consider the data structure $A^{*T} \cdot B^{T*}$ (the analysis for the other is similar) and let $v\in L_2$.	
	We go over all pairs of $2 m^{1/3-2\eps}$ neighbors on both sides ($L_1$ and $L_3$) and create the data 
	structure.
	This takes time $2 m^{1/3-2\eps} \cdot 2 m^{1/3-2\eps} \cdot O(1) = O(m^{2/3-4\eps}) \leq 
	O(m^{1-3\eps})$.
\end{proof}

Next, consider the data structures of the form $A^{MT} \cdot B^{TT} \cdot C^{TH}$ in \Cref{eq:DS-HM-tiny}.
\begin{claim}
	The amount of time it takes to create the data structures of the form $A^{MT} \cdot B^{TT} \cdot C^{TH}$ 
	in \Cref{eq:DS-HM-tiny} 
	for a vertex $v$ transitioning from $S$ to $T$ is $O(m^{2/3-\eps}) \leq O(m^{1-3\eps})$.
\end{claim}
\begin{proof}
	Consider the data structure $A^{MT} \cdot B^{TT} \cdot C^{TH}$ (the analysis for the other is similar) and 
	let $v\in L_2$.
	We go over all $2m^{1/3-2\eps}$ neighbors of $v$ in $L_1$ and also all $m^{1/3+\eps}$ vertices $w_i$ in 
	$L_4^H$. 
	We then use the number of paths between $v$ and $w_i$ through $L_3^T$ stored in 
	$B^{*T}\cdot C^{T*}$ (\Cref{eq:DS-common-tiny}) and create the data structure.	
	This gives us the desired time bound of $O(m^{2/3-\eps}) \leq O(m^{1-3\eps})$.
\end{proof}

Finally, consider the data structures $A^{HT} \cdot B^{TS} \cdot C^{SH}$ and $A^{HS} \cdot B^{ST} \cdot 
C^{TH}$ (\Cref{eq:DS-H-tiny}).
\begin{claim}
	The amount of time it takes to create the data structures $A^{HT} \cdot B^{TS} 
	\cdot C^{SH}$ and $A^{HS} \cdot B^{ST} \cdot C^{TH}$
	for a vertex $v$ transitioning from $S$ to $T$ is $O(m^{2/3-\eps}) \leq O(m^{1-3\eps})$.
\end{claim}
\begin{proof}
	Consider the data structure $A^{HT} \cdot B^{TS} \cdot C^{SH}$ (the analysis for the other is similar) and 
	let $v\in L_2$.
	We go over all $2m^{1/3-2\eps}$ neighbors of $v$ in $L_1^H$ and also all $m^{1/3+\eps}$ vertices $w_i$ 
	in $L_4^H$. 
	We then use the number of paths between $v$ and $w_i$ through $L_3^S$ stored in 
	$B^{*S}\cdot C^{S*}$ (\Cref{eq:DS-com-main}) and create the data structure.	
	This gives us the desired time bound of $O(m^{2/3-\eps}) \leq O(m^{1-3\eps})$.
\end{proof}


\section{Extending the Algorithm to \texorpdfstring{$4$}{4}-Cycles in General 
Graphs}\label{sec:cycles-general-graphs}

In this section, we extend our algorithm for counting $4$-cycles from layered graphs to general 
(simple) graphs.
\finaltheoremfalse
\AlgGeneral*

We solve this problem by creating a layered graph $G'$ by adding all the vertices $V$ to each of the 
four 
layers and then running the algorithm for layered graphs that we developed in \Cref{thm:alg-layer}.
The first thing we address is that when an edge update $(u,v)$ arrives in the graph, it corresponds 
to $4$ updates in the layered graph $G'$.
If the update is an insertion, we insert the edge in $D$ then $C$ then $B$ and then in $A$ and if it is 
a deletion, we delete edges in the reverse order.

The answer to the query of the number of new $4$-cycles in the general graph is the answer we get 
when the edge is inserted or deleted in matrix $D$ of the layered graph.
Note that even though we are counting paths in the layered graph $G'$, these correspond to 
\emph{walks} in the original graph $G$ since vertices are repeated in all the layers, so we 
might count some walks that are not paths. 
We now show that this does not happen.


\begin{claim}\label{clm:layered-to-general}
	The number of walks of length $3$ from $u \in L_1$ to $v \in L_4$ in the layered graph is equal to 
	the number of paths of length $3$ from $u$ to $v$ in the general graph.
\end{claim}
\begin{proof}
	All $3$-paths in the general graph from $u$ to $v$ exist in the constructed layered graph, so we 
	need to show that no extra paths exist.
	Consider a walk $u,x,y,v$ of length $3$ from $u$ to $v$ in the layered graph.
	We will show that all vertices on this path are distinct vertices of $V$ so this is a path in the 
	general graph.
	
	Firstly we have $x \neq u$ because we do not have self loops (edges of the form $(i,i)$).
	Also, $x \neq v$ because the edge $(u,v)$ does not exist.
	This is because if $(u,v)$ is an edge insertion then we insert it in matrix $D$ before we insert it in 
	$A,B,C$ so during the query this edge does not exist in $A,B,C$.
	If $(u,v)$ is an edge deletion then we delete it from $A,B,C$ before we query matrix $D$ so 
	during the query this edge does not exist in $A,B,C$.
	
	The same arguments show that $y \neq u$ and $y \neq v$.
	Finally, we have $x \neq y$ because we do not have self loops.
	Thus, all the vertices are distinct and $u,x,y,v$ is a path.
\end{proof}

\Cref{clm:layered-to-general} along with \Cref{thm:alg-layer} proves \Cref{thm:alg-general}.
Thus, the algorithm also works for general graphs and this also shows that the problem of 
maintaining the number of $4$-cycles is equivalent in layered and general graphs.

\subsection*{Acknowledgments}
We are very grateful to Abhibhav Garg, Prashant Gokhale and Zhiang Wu for their helpful 
conversations regarding the project and to the anonymous reviewers of PODS 2025 for their useful comments on the presentation of the paper.

\clearpage

\bibliographystyle{alpha}
\bibliography{new}

\clearpage

\appendix

\section{A Simple Algorithm for Counting 4-Cycles}\label{sec:alg-trivial}

In this section, we give a very simple algorithm to maintain the total number of 
$4$-cycles in a fully 
dynamic general graph with worst-case update time $O(n)$ where $n$ is the 
number of vertices in 
the graph.
\begin{lemma}\label{lem:alg-trivial}
	There is an algorithm to maintain the total number of $4$-cycles in a fully 
	dynamic general graph with worst-case update time $O(n)$ where $n$ is the 
	number of vertices 
	in the graph.
\end{lemma}

We will do this by counting the number of $4$-cycles through every new edge 
update.
The idea is to maintain the number of wedges i.e. two-paths between every pair of 
vertices.
During an edge update $(u,v)$ we can iterate over all the neighbors of $u$ and 
then use the 
number of wedges to $v$ stored in the data structure.
We can also update the number of wedges between all pairs of vertices in $O(n)$ 
worst-case 
update time.
We show this in the following claim.
\begin{claim}\label{clm:DS-alg-trivial}
	The number of wedges between all pairs of vertices can be maintained in $O(n)$ 
	worst-case 
	update time.
\end{claim}
\begin{proof}
	Let $(u,v)$ be the current edge update.
	We go over all the neighbors $w$ of $u$ (at most $n$) and update the number 
	of wedges 
	between $w$ and $v$.
	We do the same for $v$.
	The total time taken is $O(n)$.
\end{proof}

One important thing to note is that we need to be careful about whether we update 
the data 
structures first or answer the queries first. 
The data structures for the wedges should not have the current edge $(u,v)$ as 
part of their paths. 
To make this happen we answer the query first during an edge insertion and update 
the data 
structures first during an edge deletion.
We now show that we can answer the query of the number of $4$-cycles through 
the new edge 
update in $O(n)$ time.

\begin{claim}\label{clm:query-alg-trivial}
	Counting the number of $4$-cycles through the new edge update takes $O(n)$ 
	worst-case time.
\end{claim}
\begin{proof}
	Let $(u,v)$ be the current edge update.
	We iterate over all the neighbors (at most $n$) $w_i$ of $u$ and then use the 
	number of wedges 
	from $w_i$ to $v$ stored in the data structure.
	This gives us all the $3$-paths from $u$ to $v$ but might include some walks as 
	well.
	We now show this is not the case.  
	
	We know that $u,v,w_i$ are all distinct.
	Let $x$ be the vertex between $w_i$ and $v$.
	$x$ is different from $w_i$ and $v$ because the graph does not have self loops.
	$x$ is different from $u$ because the edge $(u,v)$ does not exist in the graph.
	This is because if the update was an edge insertion we answer the query first, so 
	the data 
	structures do not contain the edge $(u,v)$.
	Otherwise, if the update was an edge deletion we update the data structures first, 
	so they do not 
	contain the edge $(u,v)$.
\end{proof}

\Cref{clm:DS-alg-trivial} and \Cref{clm:query-alg-trivial} together prove 
\Cref{lem:alg-trivial}.
 

\section{Verifying Constraints}
\label{sec:apx-verify-constraints}
In this section, we verify that the solutions we provided for the values of $\eps$ and $\delta$ for the main 
algorithm (\Cref{sec:main-alg}) and $\e1$ and $\e2$ for the algorithm where $A,C$ are fixed 
(\Cref{sec:alg-simple}), actually satisfy all the constraints. 
We use the complexity term balancer in \cite{Complexity} to get the values of $\omega(\cdot,\cdot,\cdot)$ 
and find the optimal values of the parameters.
We start with the values of $\eps$ and $\delta$ for the main algorithm (since the values for the algorithm 
where $A,C$ are fixed, are a function of $\eps$).

\subsubsection*{Main algorithm with current best bound on the matrix multiplication exponent}
In \Cref{sec:setup-main}, we claimed that $\eps=0.0098109$ and $\delta=3\eps = 0.0294327$ satisfy all the 
constraints when $\omega=2.371339$.

\begin{itemize}
\item 
It is easy to verify that \Cref{eq:simple-constraint} is satisfied:
$\eps \leq 1/6$
since $1/6 \geq 0.16$.

\item 
It is easy to verify that \Cref{eq:iteration-pairs-constraint} is satisfied:
$3\eps \leq \delta$
since $\delta = 3\eps$.

\item 
Finally, we need to verify \Cref{eq:phase-constraint}:
$1-\delta \geq (2\omega +1) \cdot \eps + (\omega-1) 2/3.$

Substituting $\delta=3\eps$ we get:
$(6 \omega +12)\eps \leq 3- 2 (\omega -1).$

The left-hand side of the equation evaluates to $0.2573206187706$ and the right-hand side of the equation 
evaluates to $0.2573220000000003$ satisfying the equation.
Thus, all the equations are satisfied.
\end{itemize}

\subsubsection*{Main algorithm with best possible bound on the matrix multiplication exponent}

In \Cref{sec:setup-main}, we claimed that $\eps=1/24$ and $\delta=3\eps = 1/8$ satisfy all the constraints 
when $\omega=2$.

\begin{itemize}
\item 
It is easy to verify that \Cref{eq:simple-constraint} is satisfied:
$
\eps \leq 1/6
$
since $\eps=1/24$.

\item 
It is easy to verify that \Cref{eq:iteration-pairs-constraint} is satisfied:
$
3\eps \leq \delta
$
since $\delta = 3\eps$.

\item 
Finally, we need to verify \Cref{eq:phase-constraint}:
$1-\delta \geq (2\omega +1) \cdot \eps + (\omega-1) 2/3.$

The left-hand side is $7/8$ because $\delta=1/8$.
The right-hand side is $5/24+2/3=7/8$ since $\omega=2$.
Thus, all the equations are satisfied.
\end{itemize}

\subsubsection*{Algorithm where \texorpdfstring{$A,C$}{A,C} are fixed with current best bounds on 
the rectangular matrix 
multiplication exponents}

In \Cref{subsec:constraints-simple}, we claimed that $\e1 =0.04201965$ and 
$\e2=0.14568075$ when $\eps=0.0098109$.
\begin{itemize}
\item 
It is easy to verify that \Cref{constraint:L_two-simple} is satisfied: $\e1-\e2 \leq 1/3$ since $\e1-\e2 \leq 0$.

\item 
It is easy to verify that \Cref{constraint:L_one-simple} is satisfied: $\e1 \leq 1/6$ since $1/6 \geq 0.16$.

\item
It is easy to verify that \Cref{constraint:DS-L-L-simple} is satisfied: $3\e1 +2\eps \leq \e2$ since $3\e1 +2\eps 
= \e2$.
	
\item 
Consider \Cref{constraint:DS-L-HH-simple}: $\omega({2/3+2\eps},{1/3-\e1+\e2},{1/3-\e1+\e2}) \leq 4/3 -2\e1$. 
\\
The terms inside $\omega(\cdot,\cdot,\cdot)$ on left-hand side are $2/3+2\eps \leq 0.6764776$ and 
$1/3-\e1+\e2 \leq 0.436994434$.
We get $\omega({2/3+2\eps},{1/3-\e1+\e2},{1/3-\e1+\e2}) + 2\e1 \leq 1.24039952 + 2 (0.04201965) 
=1.32443882 < 1/3$. 
Thus, \Cref{constraint:DS-L-HH-simple} is satisfied.

\item 
Finally, consider \Cref{constraint:DS-H-simple}: $\omega({1/3+\e1},{2/3-\e1},{1/3+\e1}) \leq 4/3 -2\e1$. \\
The terms inside $\omega(\cdot,\cdot,\cdot)$ on left-hand side are $1/3+\e1 \leq 0.375353$ and $2/3-\e1 \leq 
0.6246471$.
We get $\omega({1/3+\e1},{2/3-\e1},{1/3+\e1}) +2\e1 \leq 1.10495201 + 2 (0.04201965) 
= 1.18899131 < 1/3$.
Therefore, all the equations are satisfied.
\end{itemize}

\subsubsection*{Algorithm where \texorpdfstring{$A,C$}{A,C} are fixed with best possible bounds on 
the rectangular matrix multiplication exponents}

In \Cref{subsec:constraints-simple}, we claimed that $\e1 = 1/24$ and 
$\e2=5/24=$ when $\eps=1/24$.

\begin{itemize}
\item 
It is easy to verify that \Cref{constraint:L_two-simple} is satisfied: $\e1-\e2 \leq 1/3$ since $\e1-\e2 \leq 0$.

\item 
It is easy to verify that \Cref{constraint:L_one-simple} is satisfied: $\e1 \leq 1/6$ since $1/6 \geq 0.16$.

\item
It is easy to verify that \Cref{constraint:DS-L-L-simple} is satisfied: $3\e1 +2\eps \leq \e2$ since $3\e1 
+2\eps = \e2$.

\item 
Consider \Cref{constraint:DS-L-HH-simple}: $\omega({2/3+2\eps},{1/3-\e1+\e2},{1/3-\e1+\e2}) \leq 4/3 -2\e1$. 
\\
$\omega({2/3+2\eps},{1/3-\e1+\e2},{1/3-\e1+\e2}) +2\e1 = 1+2\eps +\e2-\e1 +2\e1 = 1+8/24=4/3$.
	
\item 
Finally, consider \Cref{constraint:DS-H-simple}: $\omega({1/3+\e1},{2/3-\e1},{1/3+\e1}) \leq 4/3 -2\e1$. \\
$\omega({1/3+\e1},{2/3-\e1},{1/3+\e1}) +2\e1 = 1 +2\e1 < 4/3$.	
Therefore, all the equations are satisfied.
\end{itemize}

\end{document}